\documentclass[journal]{IEEEtran}
\usepackage{color}

\usepackage[colorlinks]{hyperref}
\usepackage{textcomp}
\usepackage{xcolor}
\usepackage{balance}
\usepackage{graphicx}
\usepackage[english]{babel}
\usepackage{amsthm,amssymb}
\usepackage{amsmath}
\usepackage{mathrsfs}
\usepackage{color}
\usepackage{graphicx}
\usepackage{subfigure}
\usepackage{mathtools}
\usepackage{pgf}
\usepackage{amsfonts}
\usepackage{bm}
\usepackage[bottom]{footmisc}
\usepackage{algorithm}
\usepackage{algpseudocode}
\usepackage[T1]{fontenc}
\usepackage{bbm}
\usepackage{enumerate}
\usepackage{bigints}
\usepackage{relsize}
\usepackage{balance}
\newtheorem{theorem}{Theorem}

\newcommand{\suchthat}{\;\ifnum\currentgrouptype=16 \middle\fi|\;}

\newtheorem{corr}{Corollary}
\allowdisplaybreaks
\usepackage{xcolor, soul} 

\title{Dependability Theory-based Statistical QoS Provisioning of Fluid Antenna Systems}
\author{
\IEEEauthorblockN{Irfan~Muhammad, 
                                Priyadarshi~Mukherjee, \textit{Senior Member, IEEE}, 
                                Wee Kiat New, \textit{Member, IEEE},
                                Hirley~Alves, \textit{Member, IEEE}, 
                                Ioannis~Krikidis, \textit{Fellow, IEEE}, and 
                                Kai-Kit Wong, \textit{Fellow, IEEE}}
\vspace{-6mm}

\thanks{This work was supported by 6G Flagship (Grant Number 369116) funded by the Research Council of Finland and Tauno Tönning Foundation. Preliminary results of this work have been presented at the IEEE International Workshop on Signal Processing Advances in Wireless Communication, July 2022, Oulu, Finland \cite{spawc}, where only closed form expression for the LCR over Rayleigh fading was presented.}
\thanks{I. Muhammad and H. Alves are with the Centre for Wireless Communications, University of Oulu, 90570 Oulu, Finland (e-mail: $\rm \{irfan.muhammad,Hirley.Alves\}@oulu.fi$).}
\thanks{P. Mukherjee is with the Advanced Computing \& Microelectronics Unit, Indian Statistical Institute, Kolkata, India (e-mail: $\rm priyadarshi@ieee.org$).}
\thanks{I. Krikidis is with the Department of Electrical and Computer Engineering, University of Cyprus, Nicosia 1678 (e-mail: $\rm krikidis@ucy.ac.cy$).} 
\thanks{W. K. New is with the Department of Electronic and Electrical Engineering, University College London, WC1E 7JE, London, United Kingdom (e-mail: $\rm aven.wknew@yahoo.com$).} 
\thanks{K. K. Wong is affiliated with the Department of Electronic and Electrical Engineering, University College London, Torrington Place, WC1E 7JE, United Kingdom and he is also affiliated with Yonsei Frontier Lab, Yonsei University, Seoul, Korea (e-mail: $\rm kai\text{-}kit.wong@ucl.ac.uk$).}
}
\date{July 2025}

\begin{document}

\maketitle

\begin{abstract}
Fluid antenna systems (FAS) have recently emerged as a promising technology for next-generation wireless networks, offering real-time spatial reconfiguration to enhance reliability, throughput, and energy efficiency. Nevertheless, existing studies often overlook the temporal dynamics of channel fading and their implications for mission-critical operations. In this paper, we propose a dependability-theoretic framework for statistical quality-of-service (QoS) provisioning of FAS under finite blocklength (FBL) constraints. Specifically, we derive new closed-form expressions for the level-crossing rate (LCR) and average fade duration (AFD) of an $N$-port FAS over Nakagami-$m$ fading channels. Leveraging these second-order statistics, we define two key dependability metrics such as mission reliability and mean time-to-first-failure (MTTFF), to quantify the probability of uninterrupted operation over a defined mission duration. We further extend the classical effective capacity (EC) concept to incorporate mission reliability in the FBL regime, yielding a mission EC (mEC). To capture energy efficiency under bursty traffic and latency constraints, we also develop the mission effective energy efficiency (mEEE) metric and formulate its maximization as a non-convex fractional optimization problem. This problem is then solved via a modified Dinkelbach's method with an embedded line search. Extensive simulations uncover critical trade-offs among port count, QoS exponent, signal-to-noise ratio, and mission duration, offering insights for the design of ultra-reliable, low-latency, and energy-efficient industrial internet-of-things (IIoT) systems.
\end{abstract}


\IEEEpeerreviewmaketitle

\section{Introduction}
\IEEEPARstart{S}{ixth}-generation (6G) wireless networks aim to revolutionize communications by delivering ultra-high capacity, ultra-low latency, massive connectivity, and seamless integration of sensing and artificial intelligence (AI). Amongst the key enabling technologies, the fluid antenna system (FAS) has emerged as a promising candidate \cite{KiatSurvey,wu2024fluid,Lu-2025,FAS-LLM}. Originally introduced by Wong {\em et al.}~in \cite{I22_wong2020perflim,wong2020fluid}, FAS aims to provide an integrated physical-layer system exploiting reconfigurable antennas such as liquid antennas \cite{I24_shen2024design}, reconfigurable pixels \cite{I26_zhang2024pixel}, metamaterials \cite{Liu-express2025} and etc.~for shape and position flexibility to empower the physical layer of wireless communications. This adaptability enhances system flexibility, reliability, spatial diversity, and interference resilience, making FAS particularly well-suited for Industrial Internet of Things (IIoT).

Since the first papers in \cite{I22_wong2020perflim,wong2020fluid}, the theoretical performance of FAS has been extensively studied, with \cite{Khammassi-2023} providing an eigenvalue-based channel model for characterizing the spatial correlation over the ports. Then in \cite{New-2023}, the authors analyzed the diversity order of a point-to-point FAS Rayleigh fading channel while \cite{Vega-2023} investigated the performance in Nakagami channels using a simpler spatial correlation model. Moreover, FAS in $\alpha$-$\mu$ fading channels was also examined in \cite{Alvim-2024}. The mathematical tractability in analyzing the FAS channels while maintaining accuracy was achieved by using the spatial block-correlation model in \cite{H7_Espinosa2024Anew}. Copulas were also illustrated to be useful in the performance analysis for FAS \cite{Ghadi-2023}, which works for any fading distributions. In \cite{Psomas-dec2023}, the continuous FAS was modelled and analyzed. Also, the diversity and multiplexing trade-off for multiple-input multiple-output (MIMO) channels with FAS at both ends was characterized in \cite{Kiat_TWC}. There have also been attempts to synergize FAS with other technologies such as multiuser MIMO \cite{HaoXu}, non-orthogonal multiple access \cite{KiatLetter}, integrated sensing and communications \cite{Jiaqi,H17_wang2024fluid}, over-the-air federated learning \cite{ahmadzadeh2025}, and millimeter-wave communications \cite{Leila_Globecom}. Channel estimation for FAS channels has also been studied by different approaches \cite{Skouroumounis-tcom2023,xu2024channel,New-2025ce,zhang2023successive}.

Compared to conventional communication scenarios, IIoT requires resilient connectivity, imposing stringent requirements on wireless systems \cite{mahmood2024resilience}. These systems must support ultra-high reliability, low latency, and scalable connectivity to enable mission-critical operations such as smart manufacturing, autonomous logistics, and real-time monitoring. Traditional wireless systems with fixed-position antennas often struggle to meet these demands, especially in dynamic industrial environments. FAS, with its ability to dynamically reconfigure the antenna's position in real time, offers a compelling solution by switching to spatial locations with greater signal strength. FAS can enhance resilience in the network through diversity, which in turn improves reliability and throughput.

Mission-critical IIoT applications must be resilient, meaning that they can withstand or recover from failures during their mission duration \cite{mahmood2024resilience}. One critical component of resilience is reliability, which must also be guaranteed instantaneously and over a defined operational period. This prompts the introduction of the concept of mission reliability, that is, the probability that communication remains uninterrupted throughout the mission duration~\cite{missionReliability_2018}. Closely related to mission reliability is the mean time-to-first-failure (MTTFF), a quantity that measures the expected duration before a communication breakdown or outage occurs. A high MTTFF is therefore essential for sustaining continuous operations such as closed-loop control and autonomous system coordination. Traditional wireless systems often fall short in maintaining high mission reliability and MTTFF due to their static configurations. In contrast, the dynamic adaptability of FAS enables real-time spatial reconfiguration to track favourable channel conditions, effectively extending the MTTFF and enhancing overall mission reliability.

Meanwhile, the sporadic nature of uncertain traffic arrival challenges the provision of ultra-reliable low-latency communication (URLLC) \cite{survay_traffi_Models}. URLLC applications generate traffic that demands time-varying wireless channel service assurances. This may result in quality-of-service (QoS) violations due to random environmental variation. The random nature of the wireless medium makes it crucial to guarantee the deterministic QoS for delay-sensitive services. Therefore, the authors in the seminal work \cite{ART:DAPENGwu-2003} introduced effective capacity (EC) as a link-layer model to address delay constraints. EC determines the maximum arrival rate that a random wireless channel can support (service rate) using link-level QoS metrics, such as delay outage probability and maximum delay bound, and captures physical and link-layer characteristics to ensure that statistical QoS requirements are met. Hence, EC is considered a measure of throughput supported by a random wireless channel with a specific latency constraint. Thus, EC, which combines capacity and latency, is a suitable metric for evaluating the performance of delay-constrained systems. 

EC has been widely studied in several scenarios in the last decade to examine the trade-off between reliability, security, latency, and energy efficiency (see \cite{Irfan_Secure_SC, ART:EC_survey-2019} and references therein). Unfortunately, the works related to EC only rely on the average signal-to-noise ratio (SNR), which provides limited insights into fluctuations in fading channels. The authors in~\cite{Irfan_mEC} filled this gap by utilizing the dependability theory and introducing the mission EC (mEC) metric, which not only relies on average SNR but also aids in identifying failure events in a wireless fading channel. The mEC metric helps us introduce the \emph{mission effective energy efficiency (mEEE)} by leveraging second-order statistics, such as the level crossing rate (LCR) and average fade duration (AFD).

The fading nature of wireless channel is a major hurdle to reliability and energy efficiency. Two key metrics, namely LCR and AFD, are useful to characterize such dynamic behaviour. LCR specifies the number of times the signal envelope exceeds a specified threshold level in the negative direction. At the same time, AFD quantifies the average time of fade during which the signal stays below that threshold. Both metrics are significant in analyzing the temporal characteristics of deep fades, which directly affect latency and outage events. However, from a design perspective, frequent fading events lead to higher energy consumption, resulting in a fundamental trade-off between energy use, reliability, and delay in power-constrained wireless systems. To circumvent this, the concept of effective energy efficiency (EEE) emerged, which integrates energy consumption with other key performance indicators such as reliability, delay, and throughput. 

\subsection{Contributions}
The current works on FAS and EEE lack consideration of the detrimental impact of fading in the wireless channel and do not provide insights into the time at which the first failure occurs. Inspired by this, we aim to fill the gap by leveraging the dependability theory, which helps us address this crucial factor and advance the mathematical framework in the context of EC. Our primary contributions are outlined below.

\begin{itemize}
\item We derive novel closed-form expressions for the LCR and AFD of an $N$-port FAS operating over Nakagami-$m$ fading. These results capture the impact of port count and fading severity on temporal channel variation.

\item Building on the second-order fading statistics, we revise the mission reliability metric \cite{Irfan_mEC}, which captures the probability of uninterrupted operation over a mission duration, to Nakagami-$m$ fading and FAS. 

\item We then extend the EC framework to FAS in the finite-blocklength (FBL) regime, incorporating both delay-outage constraints and mission reliability, thereby quantifying the sustainable arrival rate that guarantees failure-free operation over some finite duration $\mathrm{\Delta T}$.

\item Furthermore, we define mEEE as the ratio of mission EC to total energy consumption, refining conventional linear power models by accounting for bursty arrival patterns and idle/transmit probabilities. This novel dependability-aware metric quantifies energy efficiency up to the first failure event.

\item We formulate the mEEE maximization problem under mission reliability and latency constraints, and solve it via a modified Dinkelbach's transform coupled with a generic line search subroutine. Our approach yields an efficient iterative algorithm with provable convergence.

\item Finally, we conduct extensive Monte Carlo simulations and demonstrate key trade-offs between mission duration, reliability target, QoS exponent, antenna port count, and average SNR. The results provide actionable design guidelines for deploying FAS in ultra-reliable, low-latency, energy-constrained IIoT networks.

\end{itemize}

\noindent{\textbf{Notations}}---The probability density function (PDF) and cumulative distribution function (CDF) of a given random variable $X$ are denoted as $f_X(x)$ and $F_X(x)$, respectively, and $\mathbb{E}\left[\cdot\right]$ denotes the expectation operator. Besides, $e$ shows the Euler's number and $Q^{-1}(\cdot)$ indicates the inverse of the $Q$-function. The $I_m(\cdot)$ is the $m$th-order modified Bessel function of the first kind.  $Q_m(\cdot,\cdot)$ is the $m$th-order Marcum $Q$-function, and $\gamma(\cdot,\cdot)$, $\Gamma(\cdot,\cdot)$ are the lower and upper incomplete gamma functions \cite[6, 6.5]{BOOK:ABRAMOWITZ-DOVER03}, and $J_0(\cdot)$ refers to the zero-order Bessel function of the first kind, respectively.

 \section{System Model}\label{sc:system_model}
\begin{figure}[!t]
\centering\includegraphics[width=0.50\textwidth]{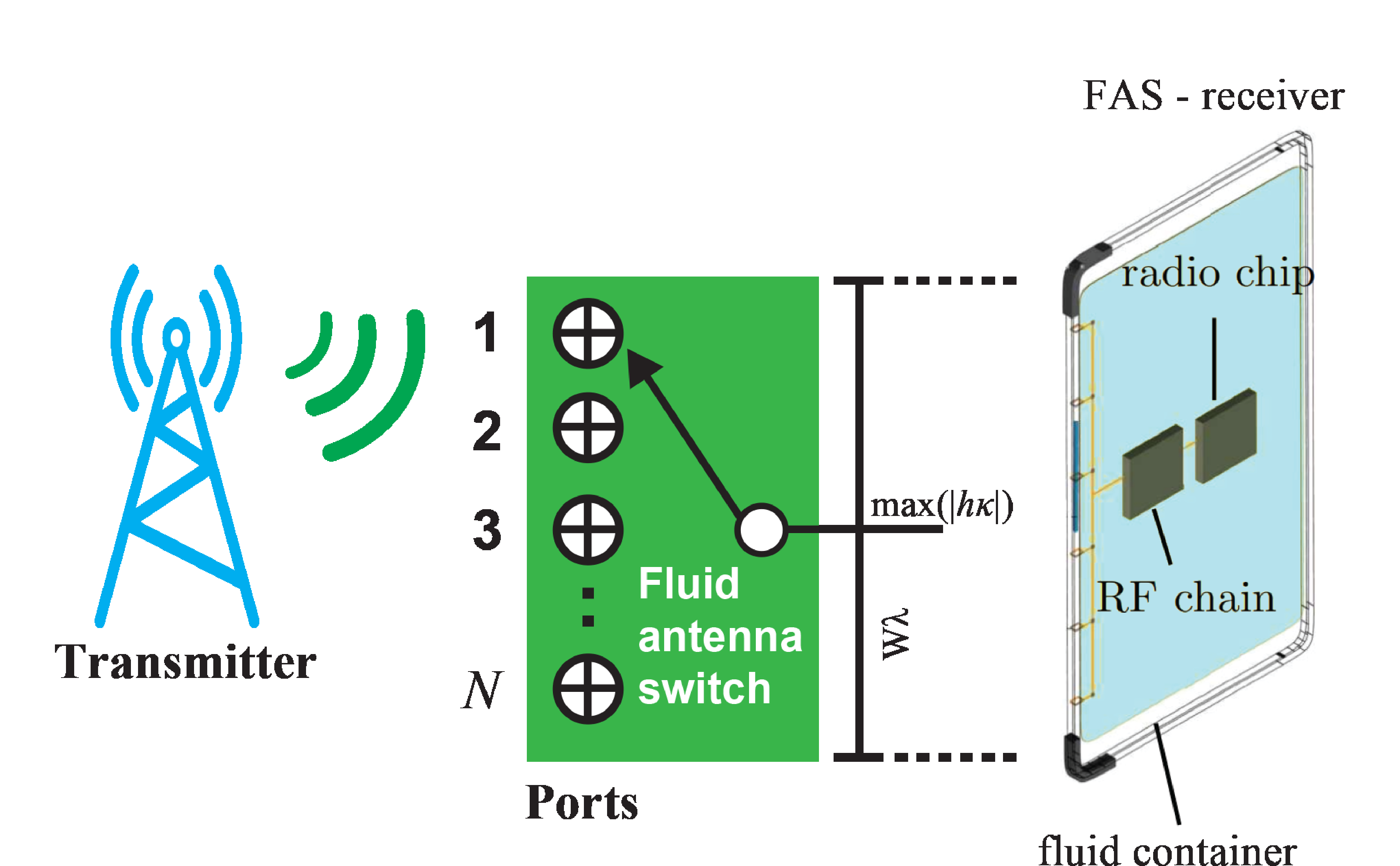}
\caption{The considered topology consists of a single fixed-position antenna transmitter and a FAS-based mobile device as receiver. In this figure, the FAS is illustrated as using a liquid-based antenna but in practice, it can be implemented using other technologies such as reconfigurable pixels.}\label{fig:model}
\vspace{-2mm}
\end{figure}

Our model consists of a simple point-to-point topology, where an energy-constrained IIoT device with a single fixed-position antenna transmits short packets to an $N$-port FAS-based common aggregator as illustrated in Fig.~\ref{fig:model}. We consider a quasi-static block fading channel, meaning that the fading coefficient remains constant within a transmission block but fluctuates independently between different fading blocks. Additionally, we consider the ideal scenario of perfect channel state information (CSI) availability at the aggregator, which establishes an upper performance limit for the system and allows us to focus our discussions on the dependability issues. Indeed, evaluating the cost of CSI at the receiver (CSIR) within the finite blocklength regime presents an interesting aspect and could serve as a potential extension of our work. Hence, the IIoT user transmits with a fixed rate due to unknown channel states. Moreover, Fig.~\ref{fig:model} illustrates that a standard FAS-based receiver operates as a position-flexible antenna system on a single radio frequency (RF) chain. The $N$ ports are uniformly distributed in a linear region of $W \lambda$, where $W \in \mathbb{R}^+$ and $\lambda$ represents the carrier wavelength. The antenna can instantly change positions among these ports~\cite{wong2020fluid}. As shown in Fig.~\ref{fig:model}, the first port serves as a reference, and the antenna is always switched to the port that offers the best channel condition.
 
 \subsection{Channel Model}
To ensure that the reliability and latency analysis remains tractable, we consider the simplified FAS channel model.\footnote{A more accurate analysis is possible with the analytical results derived in this paper by employing the block correlation model in \cite{H7_Espinosa2024Anew}. However, this extension is reserved for future work as it requires a distinct approach.} The channels at these ports are represented  as~\cite{wong2020fluid}
 \begin{equation}\label{eq1}
 \begin{cases}{h_1} = {x_0} + j{y_0} \\
 \vdots \\{h_k} = \left( {\sqrt {1 - \mu _k^2} {x_k} + {\mu _k}{x_0}} \right) \\ \qquad + j\left( {\sqrt {1 - \mu _k^2} {y_k} + {\mu _k}{y_0}} \right){\text{ for }}k = 2, \dots ,N, \end{cases} 
 \end{equation}
where $x_0,\dots,x_N$ and $y_0,\dots,y_N$ are the in-phase and quadrature components of the complex channel, respectively, and ${\mu_k}$ $\forall k$ denotes the spatial correlation among the channels \cite{wong2020fluid}
\begin{equation}\label{eq2}
\mu_k = {J_0}\left( {\frac{{2\pi (k - 1)}}{{N - 1}}W} \right),\quad {\text{for }}k = 2, \dots ,N.
\end{equation}

\begin{table}[!t]
\renewcommand{\arraystretch}{1.5}
\caption{Relevant abbreviations and symbols}\label{tab:1}
\centering
\begin{tabular}{c||l}
    \hline
    \textbf{Symbol}  &   \textbf{Definition}\\
    \hline\hline
    mEEE & Mission Effective Energy Efficiency\\
    \hline
    CSI & Channel State Information\\
    \hline
    $\mathrm{P_c}$ & Hardware Power Dissipated\\
    \hline
    $f_D$ & Doppler Frequency\\
    \hline
    $\vartheta$ & Drain Efficiency of Power Amplifier\\
    \hline
    MTTFF & Mean Time-to-First-Failure\\
    \hline
    $\mathrm{\Delta T}$ & Mission Duration\\
    \hline
    $\Upsilon$ & Failure Rate\\
    \hline
    $n$ & Blocklength \\
    \hline
    mEC & Mission Effective Capacity\\
    \hline
    $\Phi$ & Average Signal-to-Noise Ratio \\
    \hline
    $N$ & Number of Ports \\
    \hline
    $\epsilon$ & Target Error Probability \\     
    \hline
    $\mathrm{R_{M}}$ & Mission Reliability \\
    \hline
    $S$ & Source Burstiness \\
    \hline
    FBL & Finite Blocklength \\
    \hline
    R & Transmission Rate \\
    \hline
    r & Arrival Rate \\
    \hline
    $\mathrm{\overline r_{max}}$ & Maximum Average Arrival Rate\\
    \hline
    \end{tabular}
    \vspace{-4mm}
\end{table}


Thus, the received signal at the $k$-th port of the FAS is
\begin{equation}
    y_{k}=h_{k} s + w,
\end{equation}
where $h_{k}$ is the complex channel co-efficient, $s$ represents the transmitted signal with $\mathbb{E} \{ |s|^2 \}=P_t$, and $w_k$ is the additive white Guassian noise (AWGN). Moreover, with a little abuse of notations, we denote $w_k=w~\forall k$ and $w \sim \mathcal{N} \left(  0,\sigma^2_{w}\right)$. Here, we also consider a generalized Nakagami-$m$ fading scenario, i.e., $\alpha_k=|h_k|$ is a random variable following the Nakagami-$m$ distribution given by \cite{papoulis}
\begin{equation}\label{nakadefsr}
f_{\alpha_k}(\alpha)=\frac{2m^m\alpha_k^{2m-1} \exp\left(-\frac{m\alpha_k^2}{\sigma_k^2}\right)}{\Gamma(m)\sigma_k^{2m}},~ \forall \alpha \geq 0,
\end{equation}
with $\mathbb{E}\{\alpha_k^2\}=\sigma_k^2$ and $m \geq 1$ controls the severity of the amplitude fading. Furthermore, it is assumed that the FAS always chooses the port with the best channel. i.e.,
\begin{equation}\label{eq3}
\alpha_{\rm FAS} = \max \left\{ \alpha_1,\alpha_2, \ldots ,\alpha_N \right\}.
\end{equation}

Due to the proximity of the ports in an FAS, it is understood that the spatial correlation significantly influences the performance of port selection. Accordingly, the joint PDF and the joint CDF of the spatially correlated Nakagami-$m$ random variables, $\alpha_1,\dots,\alpha_N$, are, respectively, expressed as \cite{FAS_NAKAGAMI}
\begin{align}  \label{pdf_FAS_Nakagami}
&f_{\alpha_1, \dots ,\alpha_N}\left( x_1, \dots ,x_N \right) \nonumber \\
&= \frac{2x_1^{2m - 1} m^m}{\Gamma (m)\sigma^{2m}} \exp\left(- \frac{mx_1^2}{\sigma^2}  \right)  \prod_{k=2}^N \frac{2mx_1^{1 - m}x_k^m}{\sigma^2\left( 1 - \mu _k^2 \right)\mu _k^{m - 1}} \nonumber \\
& \times\! \exp\!\!\left(\!- \frac{mx_k^2 + m x_1^2\mu _k^2}{\sigma^2\left( {1 - \mu _k^2} \right)}\! \right) I_{m-1} \left[ \frac{2m \mu_k x_1 x_k}{\sigma^2\left( 1-\mu_k^2 \right)} \right]
\end{align}
and
\begin{align} \label{cdf1}
&F_{\alpha_1, \dots,\alpha_N}(x_{1}, {\dots },x_{N}) \nonumber \\
&={\rm Prob}(\alpha_1 < x_{1}, {\dots },\alpha_N < x_{N}) \nonumber \\
&=\frac{2m^m}{\Gamma (m)\sigma^{2m}} \int_0^{X_1} x_1^{2m-1} \exp\!\!\left(\!\!-\frac{mx_1^2}{\sigma^2}\right) \nonumber \\
& \times\!\!\prod_{k=2}^N\!\! \left[ 1-Q_{\rm m}\left( \sqrt{\frac{2m\mu_k^2x_1^2}{\sigma^2\left( 1-\mu_k^2 \right)}} ,\sqrt {\frac{2mX_k^2}{\sigma^2\left( 1-\mu _k^2 \right)}}  \right) \right] {\rm d}{x_1}
\end{align}
for $x_{1}, {\dots },x_{N}\ge 0$. Note that mutual coupling has no impact on an FAS if only one antenna port is activated at a time and if liquid-based fluid antennas or mechanically movable antennas are used. Hence, \eqref{pdf_FAS_Nakagami} is not a conventional $N$-variate random variable, but a product of $N$ bi-variate Nakagami-$m$ random variables. Finally, the average received SNR is given by
\begin{equation}
    \Phi=\mathbb{E} \{ \alpha_{\rm FAS}^2 \}\frac{\mathbb{E} \{ |s|^2 \}}{\sigma_w^2}=\mathbb{E} \{ \alpha_{\rm FAS}^2 \}\frac{P_t}{\sigma_w^2}.
\end{equation}

\subsection{Short Packets Transmission}
Typically, communication in an industrial wireless setup occurs in short, compact messages. Hence, in this scenario, the use of finite blocklength codewords ensures high utilization, scalability, and high reliability~\cite{zhao2021IIoT}. It also reduces the encoding and decoding time, thus ensuring real-time data transmission required for industrial applications, such as automated control and emergency response. Therefore, we mainly focus on the transmission of short packets using the foundational work of \cite{Polyanskiy}, where the authors derived the achievable rate $\mathrm{R}$ to transmit short packets, relying on the target error probability $\epsilon\in [0,1]$, SNR, and the finite blocklength $n$. The received signal is considered to be successfully decoded with the probability of $1-\epsilon$ if the received power $|x|^2$ exceeds a specific threshold $\rho$. The threshold $\rho$ is set by the sensitivity of the receiver's hardware. Thus, the considered wireless channel is viewed as a repairable component following the Gilbert-Elliott model~\cite{zorzi1995accuracy}. Accordingly, we define the two states of this model by specifying the random ``channel state'' as
\begin{equation} \label{on-off-FBL}
U(t)_{\mathrm{FBL}}= \begin{cases} 
0, & \text {if}~x_\mathrm{th} < \rho,~\mbox{``OFF'', ``failed''},\\ 
1, & \text {if } x_\mathrm{th} \geq \rho,~\mbox{``ON'', ``operational''}. 
\end{cases}
\end{equation}
Here, $\rho=\sqrt{\frac{\eta}{\Phi}}$ and $\eta$ is the solution of \cite[(18)]{Onel_2018}, which is solved recursively for a given $\epsilon$ and $\mathrm{R}$ by defining a stopping criterion $\eta_{\Delta}$, e.g., until $\lvert \eta ^{(i)}-\eta ^{(i-1)} \rvert < \eta_{\Delta} $ is met, and using $\eta^{(0)}=\infty$. In this case, the main objective is to iterate over
\begin{equation} \label{shortPacket_thresold} 
{\eta }^{(j)}=2^{\mathrm{R}+\frac {1}{\sqrt {n}} \sqrt {1-\frac {1}{(1+ {\eta }^{(i-1)})^{2}}}\log _{2}e~Q^{-1}(\epsilon)}-1, 
\end{equation}
where the superscript $j$ is the iteration index. In~\eqref{on-off-FBL}, ``operational'' implies a scenario with an error probability lower than $\epsilon$; otherwise, the system is deemed to be in a ``failed'' state. Moreover, the transition rates between these two channel states, that is, the failure and repair rates, are denoted as $\Upsilon$ and $\beta$, respectively. While the first-order statistics, such as the PDF and the CDF, are useful in understanding the overall distribution of fades in the wireless channels, they do not consider the temporal distribution of these fades. However, the second-order statistics, such as the LCR and AFD, provide insight into the channel's behaviour and ameliorate our understanding of how the channel behaves over time. Accurately defining LCR and AFD is indispensable as they directly represent the rate at which a signal varies over time. Such information is necessary for designing and analyzing communication systems; it also helps in detecting failure events in URLLC networks.

\section{Performance Metrics}\label{sc:performance metrics for Finite Blocklength}
\subsection{LCR}
Now we characterise the LCR of an FAS, which facilitates evaluating the impact of the time-varying channel on the FAS performance. The LCR of a random process $x$ at the threshold $x_{\rm th}$ essentially gives the number of times per unit duration that $x$ crosses $x_{\rm th}$ in the negative (or positive) direction \cite{rice}. Mathematically, it is defined as
\begin{equation}\label{lcrdef}
L (x_{\rm th})=\int_0^{\infty}\dot{x}f_{\dot{X}X}(\dot{x},x=x_{\rm th})d\dot{x},
\end{equation}
where $\dot{x}$ is the time derivative of $x$ and $f_{\dot{X},X}(\dot{x},x)$ is the joint PDF of $x(t)$ and $\dot{x}(t)$ in an arbitrary instant $t$. Note that, for an isotropic scattering scenario, the time derivative of the signal envelope is Gaussian distributed with zero mean, irrespective of the fading distribution \cite{stuber}. As we aim to analyze the LCR of an FAS, we present the following theorem. 

\begin{theorem}\label{theo1}
The LCR for an $N$-port FAS is given by \eqref{lcrp} (see top of next page).
\end{theorem}
\begin{figure*} [t]
\begin{align}  \label{lcrp}
&L(x_{\rm th})=\sqrt{\frac{2 \pi}{m}} {\sigma} f_D \Biggl\{ \frac{{x_{\rm th}^{2m - 1}{m^m}}}{{\Gamma (m)\sigma^{2m}}}{\exp\!\!\left(\!\! - \frac{{x_{\rm th}^2m}}{{\sigma^2}} \!\! \right)} \!\!\prod _{k=2}^N \!\!{\left[ \!{1\! - \!{Q_{\text{m}}}\!\!\left(\!\! {\sqrt {\frac{{2m\mu _k^2x_{\rm th}^2}}{{\sigma^2\left( {1 - \mu _k^2} \right)}}},\sqrt {\frac{{2m  x_{\rm th}^2}}{{\sigma^2\left( {1 - \mu _k^2} \right)}}}} \right)} \right]}\!\! + \!\!
\sum_{i=2}^N \frac{2m x_{\rm th}^m\sigma^{m - 1}\exp\!\!\left(\!\! - \frac{mx_{\rm th}^2 }{\sigma^2\left( {1 - \mu _i^2} \right)}\right) }{\sigma^{m+1}\left( 1 - \mu _i^2 \right)\mu _i^{m - 1}} \nonumber \\
&\times\!\!\bigintsss_0^{x_{\rm th}}\frac{2x_1^{3m - 2} m^m \exp\!\!\left( \!\!- \frac{m x_1^2}{\sigma^2\left( {1 - \mu _i^2} \right)}\!\right)}{\Gamma (m)\sigma^{2m}}\!I_{m-1}\!\!\left[ \frac{2m \mu_i x_1 x_{\rm th}}{\sigma^2\left( 1-\mu_i^2 \right)} \right] \!\!\prod _{\substack{k=2 \\ k \neq i}}^N \!\! {\left[ {1 - {Q_{\text{m}}}\!\!\left( {\sqrt {\frac{{2m\mu _k^2 }}{{\sigma^2\left( {1 - \mu _k^2} \right)}}}x_1 ,\sqrt {\frac{{2m}}{{\sigma^2\left( {1 - \mu _k^2} \right)}}} x_{\rm th}} \right)} \right]} dx_1\Biggr\}
\end{align}
\hrule
\end{figure*}

\begin{proof}
See Appendix \ref{app1}.
\end{proof}

We observe from \eqref{lcrp} that the LCR is a function of spatial correlation, the number of ports, the decision threshold, and the maximum Doppler frequency of the channel. This LCR of an FAS is different from that of a conventional selection combining (SC)-based receiver because of the unique PDF of the channels at the $N$ ports (see \eqref{pdf_FAS_Nakagami}) and also the aspect of the associated spatial correlation. For completeness, we consider the following two extreme cases: $\mu_k=0,1$, for $k=1,\dots,N$ and also the LCR $L(x_{\rm th})$ corresponding to $\mu_k=0~\forall k$, i.e., for a no spatial correlation scenario, given below.

\begin{corr}  \label{cor1}
For the scenario without spatial correlation, i.e., $\mu_k=0~\forall k$, $L(x_{\rm th})$ is given by
\begin{equation}\label{coin}
L(x_{\rm th})\!=\!\frac{\!\!\sqrt{2\pi} f_D N\gamma\!\left(\!m,\!\frac{m x_{\rm th}^2}{\sigma^2}\!\! \right)^{\!\!N-\!1}\!\! {x_{\rm th}}^{\!2m-\!1} \!\exp\!\!\left(\!-\frac{m x_{\rm th}^2}{\sigma^2}\! \right) }{m^{-m+\!\frac{1}{2}}\Gamma(m)^{N}\sigma^{2m-1}}. 
\end{equation}
\end{corr}

\begin{proof}
The result follows from Theorem \ref{theo1}, by replacing $\mu_k=0~\forall k$, and using $Q_1(0,b)=\int\limits_b^{\infty}xe^{-\frac{x^2}{2}}dx=e^{-\frac{b^2}{2}} \quad \text{for} \quad b \geq 0$. Without considering the aspect of spatial correlation, for a given set of system parameters, $L(x_{\rm th})$ obtained in \eqref{coin} coincides with the LCR of a conventional SC-based receiver with independent and identically distributed (i.i.d.)~channels~\cite[(15)]{iskander2001finite}. The case with $\mu_k\!=\!1~\forall k$ essentially corresponds to the scenario where the $N$ ports are identical, i.e., there is no need of any switching of the FAS among the ports.
\end{proof}

\begin{corr}\label{mu1}
For the scenario with $\mu_k\!=\!1~\forall k$, we have
\begin{equation}
L(x_{\rm th})=\frac{\sqrt{2\pi} f_D m^{m-\frac{1}{2}}{x_{\rm th}}^{2m-1} \exp\left(-\frac{m x_{\rm th}^2}{\sigma^2} \right) }{\Gamma(m)\sigma^{2m-1}}.
\end{equation}
\end{corr}

\begin{proof}
See Appendix \ref{appmu1}.
\end{proof}

It is interesting to note from the above corollary that in the case of identical channels, the FAS LCR is independent of $N$. Nevertheless, we are aware that the analytical expression of $L(x_{\rm th})$ derived in Theorem \ref{theo1}, is intricate. Consequently, we consider the simple case of $N=2$ below.

\begin{corr}\label{corr2}
For a two-port FAS, the LCR is given by
\begin{align}  \label{n2f}
L(x_{\rm th})&=\frac{2\sqrt{2\pi } m^{m + 1/2} f_D x_{\rm th}^m \exp \left(\!\!-\frac {m x_{\rm th}^2}{\sigma^{2}(1-\mu_2^2)} \right)}{ \Gamma (m)\sigma^{2m+1}  (1-\mu_2^{2}) \mu_2^{m-1}}  \nonumber  \\
& \times \sum_{k=0}^{\infty}\! \frac{\left(\mu_2 x_{\rm th}\right)^{2k+m-1}}{(k!)\Gamma(m+k) }\!\!\! \left(\!\!\frac{ m}{\sigma^2(1-\mu_2^2) }\!\!\right)^{k-1}\!\! \nonumber \\
&\times \gamma\! \left(\! k+m, \frac{m x_{\rm th}^2}{\sigma^2(1-\mu_2^2)}\! \right).
\end{align}
\end{corr}

\begin{proof}
See Appendix \ref{app2}.
\end{proof}

The above corollary demonstrates the effect of parameters such as $x_{\rm th}, \mu,$ and $\sigma$ on the LCR. For example, from \eqref{n2f} we observe that $L(x_{\rm th})$ is the product of a unimodal function and an increasing function concerning $x_{\rm th}$. This implies that $L(x_{\rm th})$ is also unimodal, that is, $L(x_{\rm th})$ initially increases with $x_{\rm th}$, but it starts to decrease after a certain point.

\subsection{AFD}\label{subsect:AFD}
The AFD $\overline r_{f}(x)$ refers to the average length of fades during which the signal envelope remains below a certain threshold \cite{stuber}. Mathematically, the AFD is expressed as 
\begin{align} \label{AFD}
\overline r_{f}(x_{\rm th})&=\frac{F_{|h_{1}|, {\dots },|h_{N}|}(x_{1}, {\dots },x_{N})}{L(x_{\rm th})}. 
\end{align}
By using~\eqref{cdf1} and~\eqref{lcrp}, the AFD for a FAS is obtained in~\eqref{AFD_FAS} (see top of next page).
\begin{figure*}[t]
\begin{equation}\label{AFD_FAS}
\overline r_{f}(x_{\rm th})=\frac{\frac{2m^m}{\Gamma (m)\sigma^{2m}} \int_0^{x_\mathrm{th}}x_1^{2m-1} \exp\left(-\frac{mx_1^2}{\sigma^2}\right)\prod_{k=2}^N\left[ 1-Q_{\rm m}\left(\sqrt{\frac{2m\mu_k^2x_1^2}{\sigma^2\left( 1-\mu_k^2 \right)}} ,\sqrt {\frac{2m x^2_\mathrm{th}}{\sigma^2\left( 1-\mu _k^2 \right)}}  \right) \right]{\rm d}{x_1}}{\sqrt{\frac{2 \pi}{m}}\sigma f_D\left\{\begin{array}{l}
\frac{{x_{\rm th}^{2m - 1}{m^m}}}{{\Gamma (m)\sigma^{2m}}}{\exp\left(-\frac{{x_{\rm th}^2m}}{{\sigma^2}}\right)}\prod_{k=2}^N \left[ {1 - {Q_{\text{m}}}\left( {\sqrt {\frac{{2m\mu _k^2x_{\rm th}^2}}{{\sigma^2\left( {1 - \mu _k^2} \right)}}} ,\sqrt {\frac{{2m}}{{\sigma^2\left( {1 - \mu _k^2} \right)}}} x_{\rm th}} \right)} \right]\\
+\sum_{i=2}^N \frac{2m x_{\rm th}^m\sigma^{m - 1}}{\sigma^{m+1}\left( 1 - \mu _i^2 \right)\mu _i^{m - 1}} \exp\left(- \frac{mx_{\rm th}^2 }{\sigma^2\left( {1 - \mu _i^2} \right)}\right)\int_0^{x_{\rm th}}\frac{2x_1^{3m - 2} m^m}{\Gamma (m)\sigma^{2m}}  \exp\left(- \frac{m x_1^2}{\sigma^2\left( {1 - \mu _i^2} \right)}\right) I_{m-1}\left[ \frac{2m \mu_i x_1 x_{\rm th}}{\sigma^2\left( 1-\mu_i^2 \right)} \right]\\
\left.\times \prod _{k=2 \atop k \neq i}^N \left[ 1 - Q_{\text{m}}\left(\sqrt{\frac{2m\mu_k^2}{\sigma^2\left(1-\mu _k^2\right)}}x_1,\sqrt{\frac{2m}{\sigma^2\left(1-\mu_k^2\right)}} x_{\rm th}\right) \right] dx_1\right\}
\end{array}
\right\}}
\end{equation}
\hrule
\end{figure*}
Unlike the conventional SC-based receivers, the AFD of an FAS takes into account the spatial correlation. Again, we consider the following two extreme cases: $\mu_k=0,1$, for $k=1,\dots,N$ as we did in the case of $L(x_{\rm th})$, and the AFD $\overline r_{f}(x_{\rm th})$ corresponding to $\mu_k=0~\forall k$, i.e., the scenario with no spatial correlation below.

\begin{corr}  \label{cor4}
In case of no spatial correlation, i.e., $\mu_k=0~\forall k$, $\overline r_{f}(x_{\rm th})$ can be expressed as
\begin{equation}  \label{AFD_cor4}
\overline r_{f}(x_{\rm th})= \frac{\gamma\left(m,m{\frac{x_{\rm th}^2}{\sigma^2}}\right)\sigma^{2m-1}}{\sqrt{2\pi} f_D N m^{m-\frac{1}{2}}{x_{\rm th}}^{2m-1} \exp\left(-\frac{m x_{\rm th}^2}{\sigma^2}  \right)}.
\end{equation}
\end{corr}

\begin{proof}
When $\mu_k=0~\forall k$, the CDF in~\eqref{cdf1} can be reduced to $\left(\!\!\frac{\gamma\left(\!m,\frac{m x_{\rm th}^2}{\sigma^2}\right)}{\Gamma{m}}\!\!\right)^{\!\!N}$. By using this reduced CDF and~\eqref{coin}, followed by some algebraic manipulation, \eqref{AFD_FAS} becomes \eqref{AFD_cor4}. 
\end{proof}

Note that the AFD obtained in~\eqref{AFD_cor4} matches with that of the conventional SC-based receiver with i.i.d.~channels~\cite[(13)]{Yacoub}.

\begin{corr}  \label{cor5}
For the case of $\mu_k=1~\forall k$, $\overline r_{f}(x_{\rm th})$ becomes
\begin{equation}  \label{AFD_cor5}
\overline r_{f}(x_{\rm th})= \frac{\gamma\left(m,m{\frac{x_{\rm th}^2}{\sigma^2}}\right)\sigma^{2m-1}}{\sqrt{2\pi} f_D m^{m-\frac{1}{2}}{x_{\rm th}}^{2m-1} \exp\left(-\frac{m x_{\rm th}^2}{\sigma^2}  \right)}.
\end{equation}
\end{corr}

\begin{proof}
Similar to LCR before, the AFD is also independent of $N$ in the case of identical channels. 
\end{proof}

The AFD obtained in~\eqref{AFD_FAS} is complicated. Therefore, we provide a simple expression for the special case $N=2$.

\begin{corr}  \label{cor6}
By substituting $N=2$ in \eqref{cdf1}, the CDF reduces to~\cite[(6.3)]{simon2008digital}
\begin{align}
&F_{|h_{1}|,|h_{2}|}(x_{1},x_{2})\nonumber\\ 
&= \frac{(1 - \mu_2^{2})^m}{\Gamma(m)}\!\! \sum_{k=0}^{\infty} \mu_2^{2k}
\frac{ \gamma\!\left(\!m \!+ \!k, \frac{mx_1^2}{\sigma^2 (1 - \mu_2^{2})}\!\right) 
       \gamma\!\left(\!m +\! k, \frac{mx_2^2}{\sigma^2 (1 - \mu_2^{2})}\right)}
     {k! \, \Gamma(m + k)}.\label{CDF_M2}
\end{align}
By taking~\eqref{AFD} into account after utilizing~\eqref{n2f} and~\eqref{CDF_M2}, the AFD for $N=2$ is given as~\eqref{AFD_M2} (see next page).
\begin{figure*} [t]
\begin{align}  \label{AFD_M2}
\overline r_{f}(x_{\rm th})=\dfrac{\dfrac{2m^m}{\Gamma (m)\sigma^{2m}} \bigintss\limits_0^{  x_\mathrm{th}    } x_1^{2m-1} \exp\left(-\dfrac{mx_1^2}{\sigma^2}\right)   \left[ 1-Q_{\rm m}\left( \sqrt{\dfrac{2m\mu_k^2x_1^2}{\sigma^2\left( 1-\mu_k^2 \right)}} ,\sqrt {\dfrac{2m x^2_\mathrm{th}  }{\sigma^2\left( 1-\mu _k^2 \right)}}  \right) \right] {\rm d}{x_1}. }
{\dfrac{2\sqrt{2\pi}  m^{m + 1/2} f_D x_{\rm th}^m \exp \left(\!\!-\frac {m x_{\rm th}^2}{\sigma ^{2}(1-\mu_2^2)} \right)}{ \Gamma (m)\sigma^{2m+1}  (1-\mu_2^{2}) \mu_2^{m-1}}  \mathlarger{\mathlarger{\sum}}_{k=0}^{\infty}\! \dfrac{\left(\mu_2 x_{\rm th}\right)^{2k+m-1}}{(k!)\Gamma(m+k) }\!\! \left(\!\!\dfrac{ m }{\sigma^2(1-\mu_2^2)}\!\!\right)^{k-1}\!\!
\gamma\! \left(\! k+m, \dfrac{m x_{\rm th}^2}{\sigma^2(1-\mu_2^2)}\! \right)}
\end{align}
\hrule
\end{figure*}
\end{corr}

Moreover, the level crossing analysis determines the average non-fade duration $\overline r_{f}(x_{\rm th})$ of the system, which is obtained by utilizing the following relationship~\cite{ohmann2015minimum}
\begin{equation}\label{ANFD}
\overline r_{n}(x_{\rm th})=\frac{1}{L(x_{\rm th})}-\overline r_{f}(x_{\rm th}).
\end{equation}
Furthermore, the reciprocals of $\overline r_{f}(x_{\rm th})$ and $\overline r_{n}(x_{\rm th})$ characterize the failure rate and repair rate as $\Upsilon=\frac{1}{\overline r_{n}(\sqrt{\eta})}$ and $\beta=\frac{1}{\overline r_{f}(\sqrt{\eta})}$~\cite{Irfan_mEC}. Next, by considering the failure rate, we examine the mission reliability, a new metric of dependability, while utilizing its attributes such as MTTFF and failure rate. 

\subsection{Mission Reliability}
The existing reliability concept for the upcoming generation of wireless communication does not explicitly include temporal features, rendering it insufficient when applied to practical time-varying wireless channels. The conventional definition of reliability is well harmonized with the notion of availability when the principles of dependability theory are applied~\cite {missionReliability_2018}. The new concept of reliability diverges from the 3GPP reliability definition, which is based on the mean time required to deliver a packet successfully and aligns well with steady-state availability, whereas mission reliability incorporates the time dependency of the system. Moreover, the work in~\cite{missionReliability_2018} thoroughly elucidates the fundamental difference between mission reliability and channel availability, demonstrating that instantaneous channel availability quickly reaches steady-state availability $A(t) = 0.99$, which indicates a $1\%$ outage probability at any given moment. In contrast, mission reliability decreases significantly and approaches zero due to the random fading process, which hinders practically failure-free operation over longer durations ($\mathrm{\Delta T}$) and necessitates improvements in mission reliability and MTTFF within the scope of URLLC. Therefore, the authors in~\cite{Irfan_mEC} established a mission reliability metric specifically designed for FBL, defining it as the probability of achieving failure-free operation during the mission duration, which vitalizes the connection between short packet transmission and reliability guaranteed operation over the mission duration, i.e.,
\begin{equation} 
\mathrm{R_{M}(\Delta T},n,\epsilon) = \Pr \left \lbrace{ U(\tau)_{\mathrm{FBL}}=1 , \forall \, \tau \in [0,\mathrm{\Delta T}] }\right \rbrace.
\end{equation}

\subsection{MTTFF}
MTTFF is an important dependability metric that quantifies the average duration a functional component of a system operates until it experiences its first failure. This metric is closely associated with mission reliability, which is connected to the time dimension. Hence, MTTFF is essential to ensure the high reliability required for mission-critical applications, such as remote surgery and industrial automation, where even short failures can result in severe repercussions. It enables improved planning and proactive maintenance of a system by providing insights into how long a system can survive until the initial failure. Although MTTFF holds significant importance in several fields, it has not garnered enough attention in the context of wireless communication.
Nevertheless, it has been recently applied to wireless communication~\cite{missionReliability_2018,Irfan_mEC,hossler2017applying}. Therefore, due to its proven efficacy in time-based reliability analysis for URLLC, we utilize it in our scenarios. 

MTTFF is obtained by using an infinite integral, or Laplace transform $\mathrm{R_{M}^{*}(s)}$ of the mission reliability function, while setting the Laplace parameter $s=0$, i.e.,
\begin{equation} {\mathrm{MTTFF}}=\int _{0}^\infty \mathrm{R_{M}(\tau)} \,\mathrm {d}\tau =\mathrm{R_{M}^{*}(0)}. \end{equation}
Our system consists of an FAS-based receiver, which selects the port with the best channel. Hence, it is considered a single component. Therefore, the closed-form expression for the MTTFF of a single component is given by~\cite{hoyland2009system_RELIABILITY}
\begin{equation} \label{MTTFF_M}
{\mathrm {MTTFF} }=\frac{1}{\Upsilon}.
\end{equation}
Thus, the mission reliability for the system is determined by leveraging MTTFF under the assumption that the time to failure with constant $\Upsilon$ is distributed exponentially~\cite{hoyland2009system_RELIABILITY}, i.e.,
\begin{equation}  \label{missionReliability_FBL}
\mathrm{R_{M}}(\mathrm{\Delta T})=\exp {\left (-{\frac {\mathrm{\Delta T}}{ {\mathrm{MTTFF}}  } }\right)}.
\end{equation}
 

\section{mEEE}\label{sc:mEEE_FAS}
Here, we evaluate the mEEE, a new dependability metric defined as the ratio between mEC and total energy consumption of the system. The expression for mEC is given, along with an analysis of the Markovian arrival model and an investigation of the effective bandwidth (EB) under statistical QoS constraint to capture the impact of bursty traffic and latency. Then we formulate the mEEE maximization problem, considering the constraints of mission reliability and latency.

\subsection{Preliminaries on Statistical QoS Provisioning}
We employ the well-known large-deviations framework of EC and its dual, the EB, to capture statistical QoS constraints at the link layer. Denoting by $\theta>0$ the QoS exponent (larger~$\theta$ implies stricter delay guarantees), the EC is defined as the maximum constant arrival rate supported by the service process $S[t]$ under exponent~$\theta$~\cite{ART:DAPENGwu-2003,Fahad_EEE_2022}:
\begin{equation}\label{eq:EC-general}
  \mathrm{EC}(\theta)
  = -\lim_{t\to\infty}\tfrac{1}{\theta t}
    \ln\mathbb{E}\bigl[e^{-\theta S[t]}\bigr]
  = -\tfrac{\Lambda(-\theta)}{\theta},
\end{equation}
where $\Lambda(\cdot)$ is the log-moment generating function of the per-block service rate (e.g., the maximal FAS rate~$R$ in \eqref{eq:EC-general}). In the two-state (ON-OFF) model, this reduces to 
\begin{multline}\label{eq:EC-ONOFF}
\mathrm{EC}(\theta)
= -\frac{1}{\theta}
    \ln\!\left[\tfrac12(V_{11}+V_{22}e^{\theta R}) \right.\\
\left.\qquad+\tfrac12\sqrt{(V_{11}+V_{22}e^{\theta R})^2
      +4(V_{11}+V_{22}-1)\,e^{\theta R}}\right],
\end{multline}
with transition probabilities $V_{ij}$. Then the EB of a discrete-time Markov arrival process with transition matrix $[\mathrm{p}_{ij}]$,
\begin{equation}\label{eq:EB-DTMC}
  a^*(\theta,r)
  = \frac{1}{\theta}\ln\Bigl[\frac{(1-S)+S\,e^{\,r\theta}}{1}\Bigr]
  = \frac{1}{\theta}\ln\bigl[1 - S + S e^{r\theta}\bigr],
\end{equation}
gives the minimal service rate guaranteeing $\Pr\{D>d_{\max}\}\!\approx\!
e^{-\theta\,a^*(\theta)\,d_{\max}}$, where $\mathrm{p_{11}}=1-S$ and $\mathrm{p_{22}}=S$. Hence, $\mathrm{P_{ON}}=S$ shows the variation of arrival rates $r$ at the source, and $\mathrm{P_{ON}}=1$ indicates constant traffic arrivals \cite{Fahad_EEE_2022,HirleySecureEC}. Under statistical QoS provisioning, buffer overflow and delay violation probabilities both decay exponentially in their respective thresholds, with exponents $\theta$ and $\theta a^*(\theta)$. Moreover, the maximum average arrival rate $\mathrm{\overline r_{max}}$ is defined as the product of the probability of arrival and average arrival rate $\mathrm{\overline r_{max}}=r \mathrm{P_{ON}}$. Furthermore, we utilize the equality $a(\theta, r)=\mathrm{m{EC}}$ to get $r$, which can support failure-free transmission for given $\Phi, \mathrm{R}, \mathrm{\Delta T}$, and $\theta$. Similarly, for the discrete-time Markov source, the $\mathrm{\overline r_{max}}$ is obtained as
\begin{align}\label{eq_MaxArrivalRate}
&\mathrm{\overline r_{max}} =\frac{S}{\theta}\ln\left(\frac{1}{S}e^{\theta\mathrm{m{EC}}}-\left(1-S\right)\right).
\end{align}

\subsection{mEEE for FBL}\label{mEEE_fbl}
The FBL notably reduces latency, boosts reliability, and lowers energy consumption in IIoT networks. Moreover, short-packet transmission reduces risk, improving functional integrity, which is critical for maintaining high data integrity in a noisy industrial environment~\cite{Vitturi}. It further enhances the battery longevity of IIoT devices by minimizing transmission time and power consumption per packet~\cite{mao2021energy}. Additionally, short packets in FBL accelerate transmission and processing, which is crucial in time-critical industrial applications where even minor delays can cause detrimental outcomes. Various industrial applications utilize short packets and require a transmission period during which no failures occur. However, traditional EEE metrics do not provide insights into such failures. Therefore, we introduce and analyze mEEE, a novel dependability theory metric that can capture the temporal characteristics of wireless fading channels. The mEEE for the system, under delay outage probability and mission duration, can be defined as the ratio between mEC and the total power consumption at the transmitter, i.e.,
\begin{align}\label{ObjFunction}
\mathrm{mEEE}(\Phi)=\frac{\mathrm{mEC}(\Phi)}{\mathrm{P_{t}(\Phi)}}.
\end{align}
It provides information about the energy efficiency of the system until the first failure occurs. The energy consumption of a system is primarily comprised of a frequency synthesiser, low-noise amplifier (LNA), filters, power amplifier (PA), digital-to-analogue (AD/DA), analogue-to-digital converters, and mixers.

The $\mathrm{P_{t}(\Phi)}=\vartheta \Phi +\mathrm{P_{c}}$ in the basic linear power consumption model of EEE is the total power consumption, where $\vartheta$ and $\mathrm{P_{c}}$ indicate the drain efficiency of PA and dissipated hardware power, respectively. But this model has the disadvantage of assuming that data is always available at the transmitter, which leads to an overestimation of energy consumption. This assumption is implausible given that continuous transmission utilizes more energy than intermittent or sporadic transmission. A more realistic approach was proposed in~\cite{Fahad_EEE_2022} based on the probability of data arrival at the source and the probability of a non-empty buffer. Thus, the fundamental linear consumption model can be refined by integrating the transmit mode with transmission probability, i.e., $\mathrm{P_{ptx}}$ (the probability of being in transmission mode) and $\mathrm{P_{idle}}$ (the probability of being in idle mode). Note that $\mathrm{P_{idle}}$ relies on two mutually independent events; the first event occurs when no traffic is generated and the subsequent event refers to when the buffer is empty. Hence, $\mathrm{P_{idle}}$ is expressed as product of probability of no data arrival $(1-\mathrm{S})$ and probability of empty buffer $(1-\zeta)$. Therefore, the refined power consumption model gives \cite{Fahad_EEE_2022}
\begin{equation} \label{refined_Pt}
\mathrm{P_{t}(\Phi)}= \vartheta  \Phi -\left ({\vartheta \Phi -\xi_{\mathrm{idl}} }\right)(1-\mathrm{S}) \left ({1-\frac {\mathrm{\overline r_{max}}}{\mathrm{R}} }\right)+\mathrm{P_{c}},
\end{equation}
where $\zeta=\frac {\mathrm{\overline r_{max}}}{\mathrm{R}}$ and $\xi_{\mathrm{idl}}$ is the constant circuit power consumption. In what follows, \eqref{ObjFunction} becomes
\begin{align}\label{ObjFunction_updated}
\mathrm{mEEE}(\Phi)
&=\frac{-\frac{1}{n \theta}\text{ln}\!\left[\!1\!-\!\exp\!\left(\!-\frac{\mathrm{\Delta T}}{\mathrm{MTTFF}} \! \right)\!\bigg(\!1\!-\!\exp\!\left(\!-\theta n \mathrm{R}\right)\!\bigg)\! \right]}{\vartheta  \Phi -\left ({\vartheta \Phi -\xi_{\mathrm{idl}} }\right)(1-\mathrm{S}) \left (1-\frac {\mathrm{\overline r_{max}}}{\mathrm{R}} \right)+\mathrm{P_{c}}}.
\end{align}

Later, we will see that $\mathrm{mEEE}$ reaches a maximum at a certain $\Phi$ and then diminishes, meaning that operating in the high SNR region wastes power without improving efficiency. Therefore, maximization of $\mathrm{mEEE}$ is of great importance in IIoT to achieve maximum throughput per unit of power consumed. The mEEE metric is valuable for assessing the efficiency of the power consumption and energy savings of the IoT system. Therefore, we propose maximizing mEEE while adhering to the mission reliability constraint concerning $\Phi$.

Fig.~\ref{mEEEvsPhi} shows the mEEE results against $\Phi$ for the fixed values of $W=0.3$, and $m=2$. The results reveal that mEEE degrades at high SNRs. The higher number of ports enhances the system throughput compared to a single-port setup. Moreover, our system consumes less power and becomes more energy-efficient when a high number of ports are utilized, particularly at a low average SNR range. Furthermore, mEEE decreases after reaching the optimal point as the SNR increases. This illustrates how optimizing the mEEE concerning $\Phi$ can enhance system performance under various parameter settings. Therefore, we design the optimization problem of maximizing mEEE under mission reliability constraints.

\begin{figure} [!t] 
\centering 
\includegraphics[width=\columnwidth]{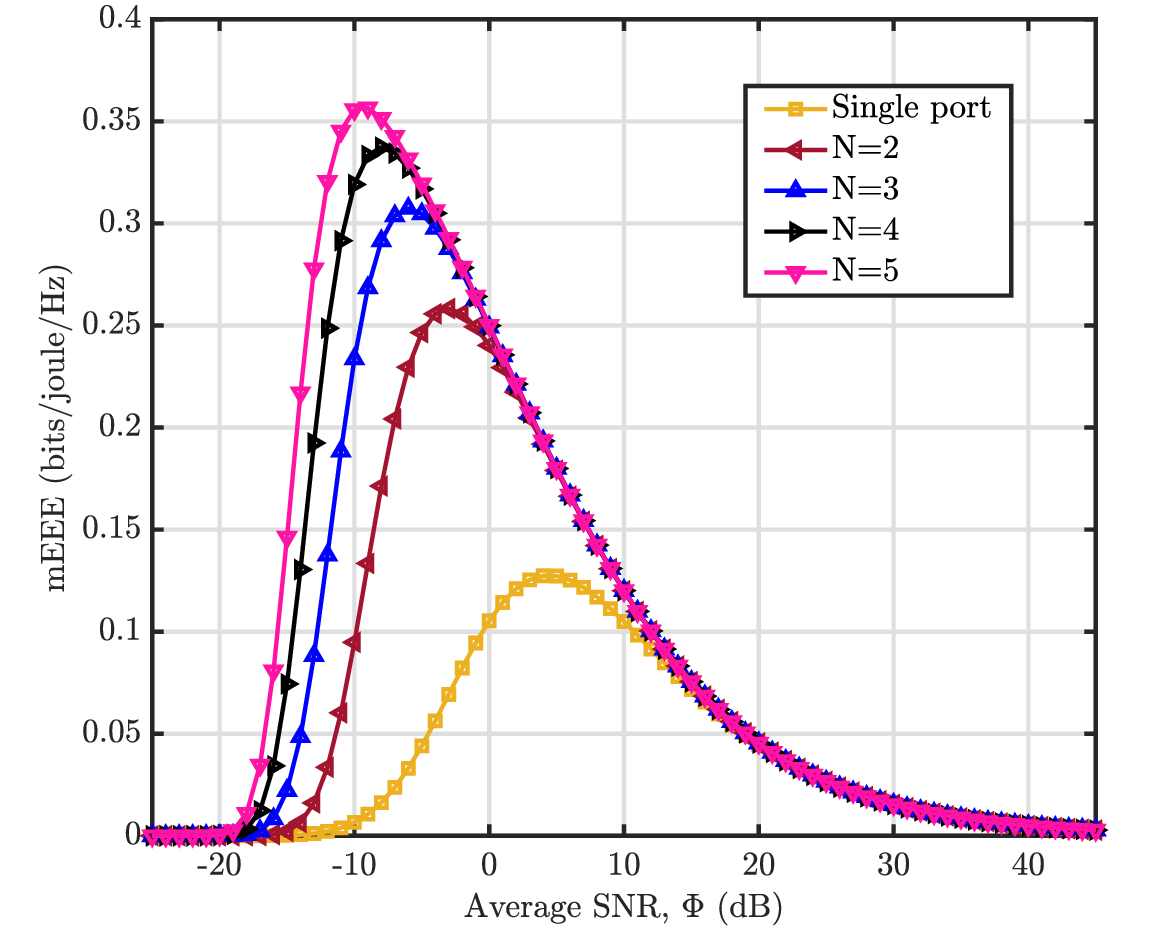}
\caption{The mEEE results as a function of $\Phi$ for different number of ports, $N$ assuming fixed $m$, $W$, and $\mathrm{\Delta T}$.}\label{mEEEvsPhi}
\end{figure}

\subsection{mEEE Maximization}
Ensuring the IIoT network's QoS and mEEE demands efficient resource utilization, such as rate and power transmission. Therefore, we focus on the power allocation strategy, and the design optimization problem is written as
\begin{equation}\label{Optimization_prob}
\arg\max_\Phi~\mathrm{mEEE}=\frac{\mathrm{mEC}(\Phi)}{\mathrm{P_{t}(\Phi)}}~~\mbox{s.t.}~~\mathrm{{R}_{M}(\Phi)}\geq \omega.
\end{equation}

Fig.~\ref{mEEEvsPhi} illustrates the quasi-concavity of~\eqref{ObjFunction_updated}, and that the mission reliability constraint is neither concave nor convex, which makes the optimization problem non-convex~\cite{boyd2004convex}. Such optimization problems are intricate, and finding a closed-form solution is intractable. Fortunately, fractional programming, specifically the popular Dinkelbach algorithm, has been extensively used to tackle these optimization problems. 

For this reason, we reformulate the single-ratio problem~\eqref{Optimization_prob} using Dinkelbach's transform as
\begin{equation}\label{Dinkel_Obj}
\arg\max_\Phi~\mathrm{mEC}(\Phi)- \kappa \mathrm{P_{t}(\Phi)}~~\mbox{s.t.}~~\mathrm{{R}_{M}(\Phi)}\geq \omega,
\end{equation}
where $\kappa$  is an auxiliary variable, and updated iteratively by 
\begin{equation}\label{iterative_update}
\kappa(i+1)= \frac{\mathrm{m{EC}}(\Phi[i])}{\mathrm{P_{t}}(\Phi[i])}, 
\end{equation}
where $i$ represents the iteration index. Convergence can be guaranteed by updating $\kappa$ alternatively according to~\eqref{iterative_update} for $\Phi$ in~\eqref{Dinkel_Obj}~\cite{shen_DK}. Dinkelbach's method depends on solving an inner loop optimization problem at each iteration~\cite{FP}. Therefore, we solve the inner loop by the golden search method to improve the convergence of Dinkelbach's method as in~\cite{Fahad_EEE_2022}.

 \begin{algorithm}[!t] \label{sc:alog1}
\caption{Modified Dinkelbach's (DK) Algorithm }
\label{CHalgorithm1}
\begin{algorithmic}[1]
\State $\mathbf{Initialization:}$ $lb$, $ub$ are the lower and upper-bounds, tolerance $\tau_1$,$\tau_2$, $q \leftarrow 0$, $i \leftarrow 0$, optimized $\leftarrow$ false

\State $\mathbf{Define:} {f(x,q)=\mathrm{\mathbf{P4}}}$
\While{(optimized $\leftarrow$ false and max-iterations )}

\State {Compute} $x_1=ub-(ub-lb)*0.618$
\State {Compute} $x_2=lb+(ub-lb)*0.618$
\State {Compute} $g_1=f(x_1,q)$
\State {Compute} $g_2=f(x_2,q)$
\While{( error $\geq$ $\tau_1$ )}
\If{$g_1 > g_2 $ }
\State {$ub \leftarrow x_2; x_2 \leftarrow x_1; g_2 \leftarrow g_1$}
\State {Compute} $x_1=ub-(ub-lb)*0.618$
\State {Compute} $g_1=f(x_1,q)$

\ElsIf{$g_1 < g_2$ }
\State {$lb \leftarrow x_1; x_1 \leftarrow x_2; g_1 \leftarrow g_2$}
\State {Compute} $x_2=lb-(ub-lb)*0.618$
\State {Compute} $g_2=f(x_2,q)$
\EndIf
\State {Compute} error $=2*\frac{ub-lb}{ub+lb}$
\EndWhile
\State {Compute} $x =\frac{x_1+x_2}{2}$ and reset $lb$, $ub$
\If{$f(x,q)=0 $}
\State {$x^* \leftarrow x$}
\State {optimized $\leftarrow$ true}
\ElsIf{$f(x,q)\leq$ $\tau_2$ }
\State {$x^* \leftarrow x$}
\State {optimized $\leftarrow$ true}

\Else
\State {update, $\ q \leftarrow \frac{f_{1}(x)}{f_2(x)}$}
\State {update, \ $ i \leftarrow i+1$}
\EndIf
\EndWhile
\end{algorithmic}
\end{algorithm}

\section{Numerical Results}\label{sc:numerical Analysis}
In this section, we provide the numerical results of mEEE for our system consisting of an FAS-aided receiver with Nakagami fading in the FBL regime by varying its different parameters such as $\theta$, $W$, $\Phi$ and $\mathrm{\Delta T}$. We consider unit power channels, i.e., $\mathbb{E}[|h_k|^2]=\sigma^2=1$ for $k=1,\dots,N$, where $N$ is the number of ports in the FAS. Unless stated otherwise, the remaining system parameters are fixed as $\mathrm{P_{c}}= 0.2$, $\mathrm{P_{idle}}=0.03$, $\eta_{\Delta}=10^{-4}$, $n=1000$ channel uses, $\vartheta = 0.2$, $\theta= 10^{-3}$, $S=0.5$, $\mathrm{R}=0.1$ bcpu and $\epsilon=10^{-2}$ throughout the simulations. The envelope threshold is set to $\rho$.

Fig.~\ref{NLCR} shows that our analytical expression for the normalized LCR (NLCR), i.e., LCR/$f_D$, is corroborated by extensive Monte Carlo simulations, indicating its correctness. In this figure, we demonstrate the NLCR against the mean SNR $\Phi$ for different variations of $N$ and $W$ for $m=1$, and $\mathrm{R}=1$. Note that the values considered are only for illustrative purposes. The figure validates our assertion that LCR is not independent of $\mu$ and $N$, as illustrated in~\cite[(7)]{I22_wong2020perflim}. Moreover, it is evident that the NLCR decreases with an increase in average SNR, and the choice of $W$ also has an impact on the NLCR.
 
 \begin{figure} [!t] 
\centering \includegraphics[width=\columnwidth]{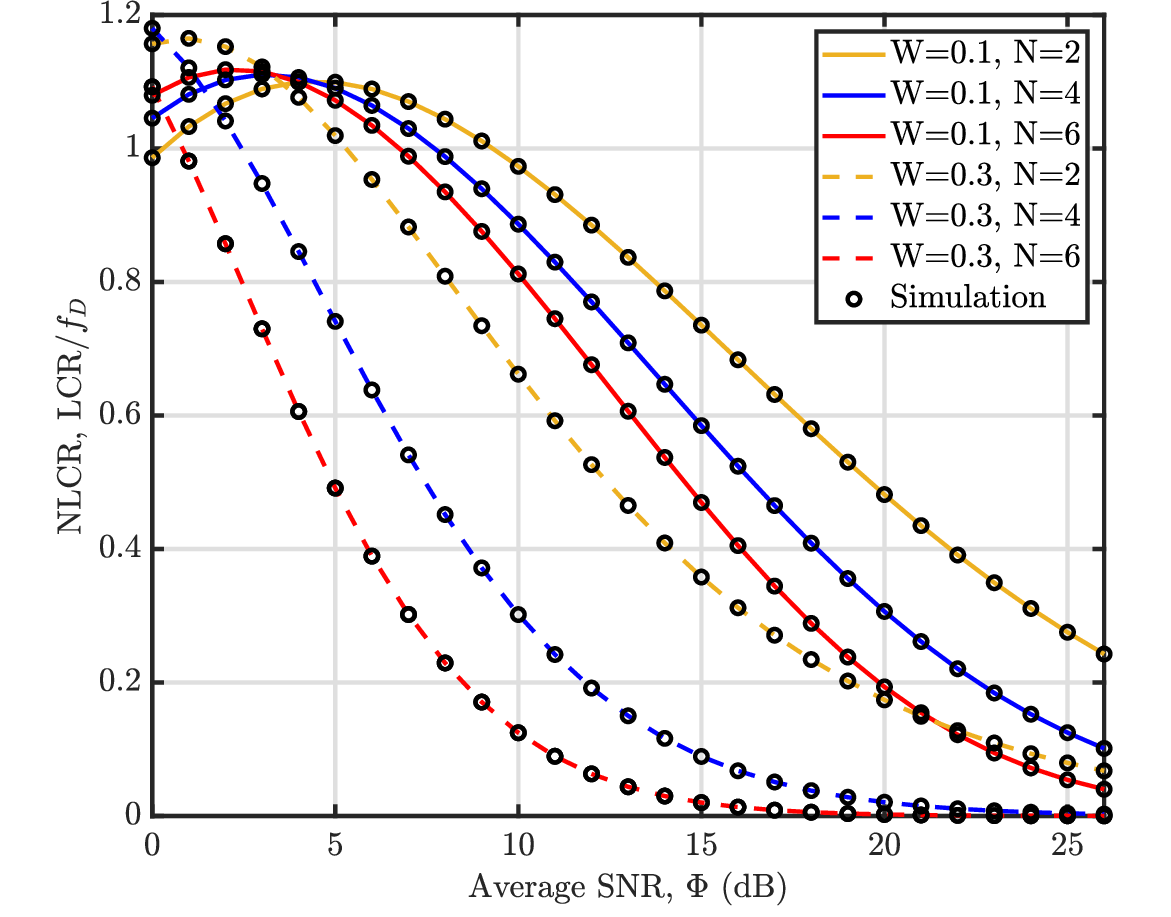}
\caption{The NLCR results against the average SNR, $\Phi$, with $m=1$ (Rayleigh fading) for different configurations of $N$ and $W$.}\label{NLCR}
\end{figure}

In Fig.~\ref{mReliabilityVsMissionDuration}, we assess the trade-off between mission reliability $\mathrm{R_{M}}(\mathrm{\Delta T})$ and mission duration $\mathrm{\Delta T}$. As we can see, mission reliability declines with increasing mission duration due to the requirement to maintain high reliability consistently over a longer period. Higher mission reliability can be achieved with a larger value of $W$. Note that we must compromise the mission reliability to have a failure-free operation for a longer mission duration. Moreover, a larger number of ports can enhance mission reliability even with a long mission duration. Next, we evaluate the optimized mEEE against mission duration, delay component, and target mission reliability. 

\begin{figure} [!t] 
\centering 
\includegraphics[width=\columnwidth]{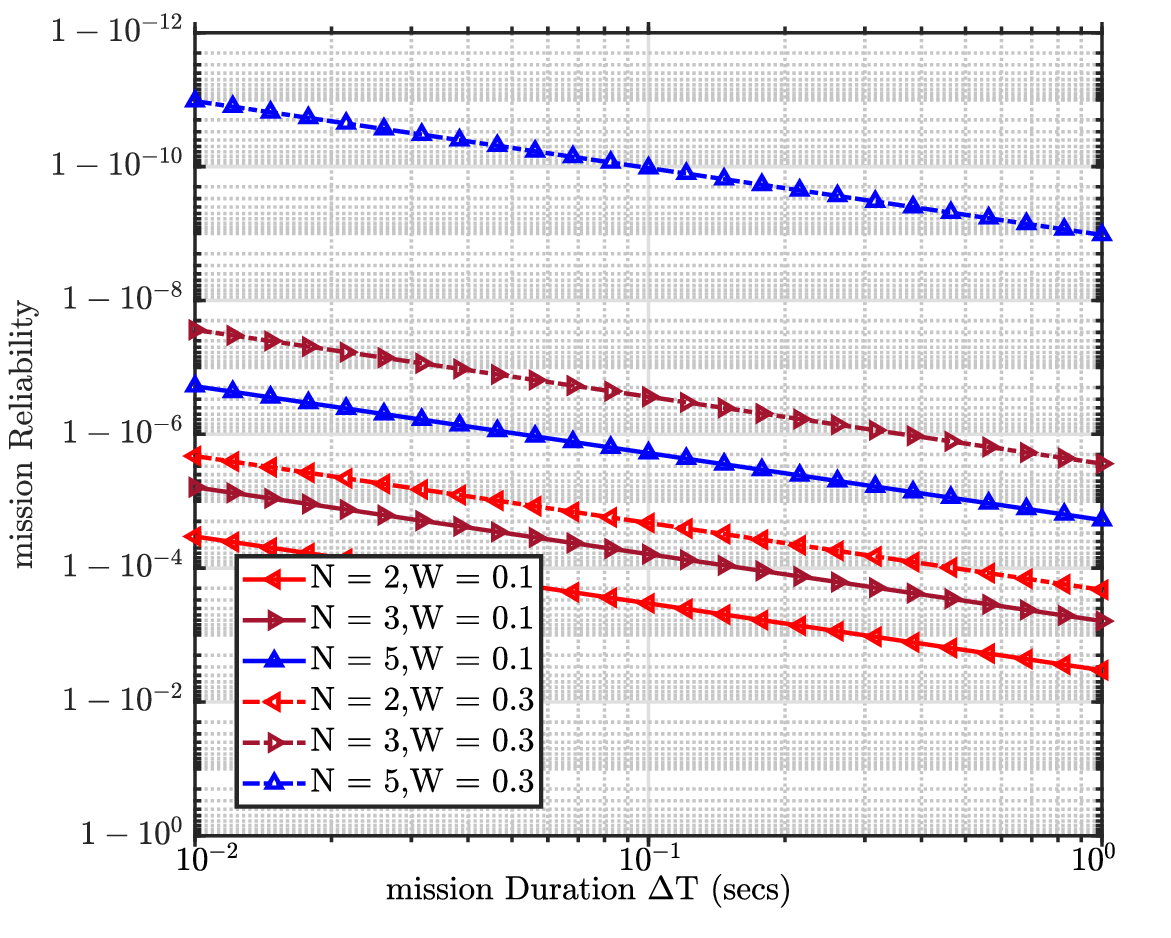}
\caption{Trade-off between $\mathrm{\Delta T}$ and $\mathrm{{R}_{M}(\Delta T)}$ under FBL for different configurations of $W$, and $N$ for $m=2$ (which corresponds to the cases with some line-of-sight (LoS) component in the channel).}\label{mReliabilityVsMissionDuration}
\end{figure}

Fig.~\ref{OptimzedmEEEVsMissionDuration} illustrates the optimized mEEE against the mission duration $\mathrm{\Delta T}$ for $m=5$, and $W=0.03$. We shed light on how mEEE varies with $\mathrm{\Delta T}$ if the target mission reliability and the QoS constraint are set to $99.99\%$ and $10^{-3}$ under optimal $\Phi$. The mEEE reduces with the enhancement of mission duration due to the fading nature of the wireless channel, since a persistent failure-free operation is difficult to maintain in a fading channel. The environment with less ports consumes more energy, making the system less efficient in terms of energy under a longer mission duration. The higher number of ports is always favourable regarding energy efficiency, ultra-reliability, and longer mission duration. It indicates that it is paramount for IIoT applications, such as long-term industrial monitoring or control systems, to consider the environment in which they operate while utilizing energy efficiently under longer mission durations, along with ultra-reliability.

\begin{figure} [!t] 
\centering 
\includegraphics[width=\columnwidth]{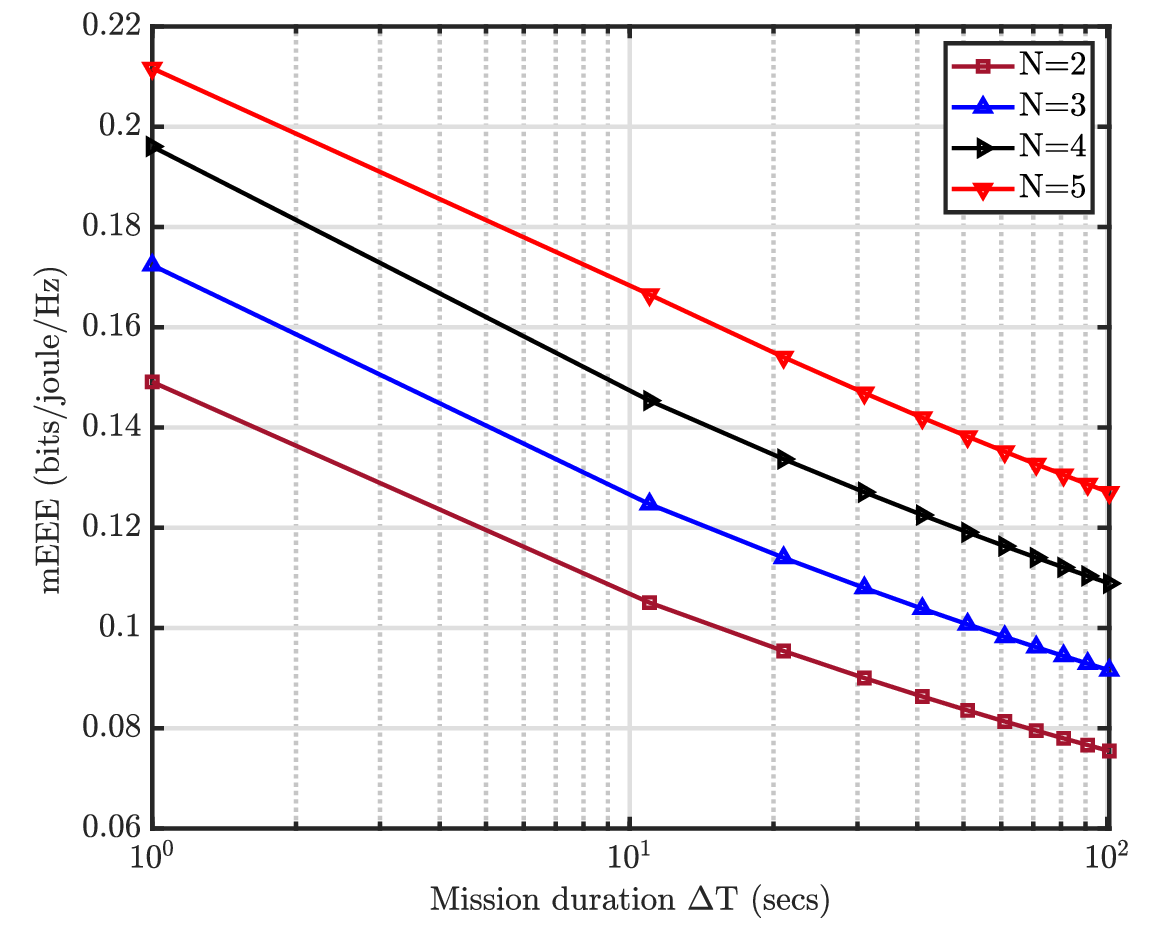}
\caption{The optimized mEEE against $\mathrm{\Delta T}$ with $W=0.03$ and $m=5$.}\label{OptimzedmEEEVsMissionDuration}
\end{figure}

In Fig.~\ref{OptimzedmEEEVsTheta}, we study the impact of QoS constraint $\theta$ on the optimized mEEE over Nakagami fading in the FBL regime. The mission reliability indicator $\omega$ is set to $99.99\%$ and mission duration $\mathrm{\Delta T}=5$, which guarantees the throughout $99.99\%$ reliable communication within a specific mission duration. We see that all $N$ divulge the decrement in mEEE under strict buffer constraint. This is because the larger $\theta$ corresponds to stringent constraints, and smaller $\theta$ indicates looser constraints, which implies the system can tolerate larger delays. The system needs to be highly energy-efficient with strict delay requirements, but it becomes more energy-efficient with a loose QoS constraint. This trade-off indicates that the system should not be operated with stringent QoS requirements.

\begin{figure} [!t] 
\centering 
\includegraphics[width=\columnwidth]{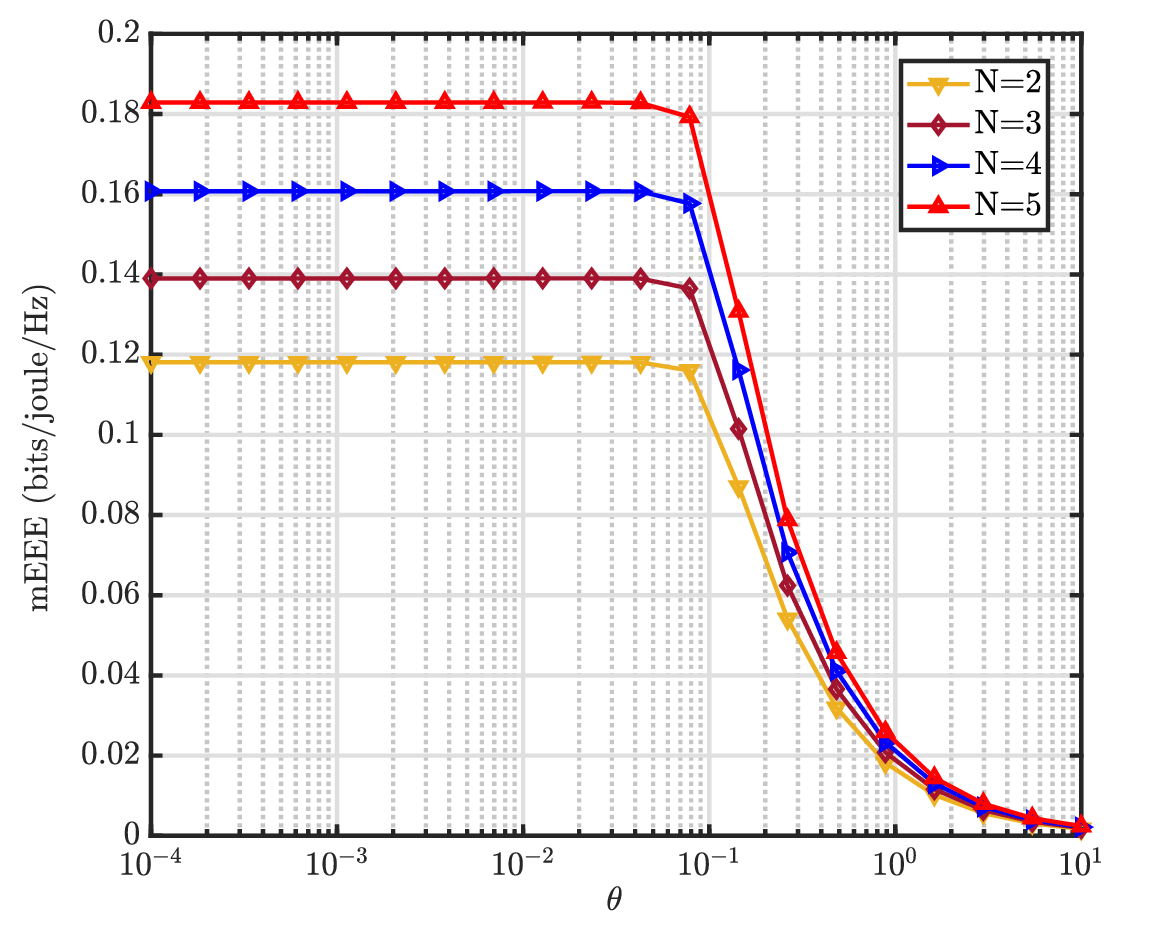}
\caption{The optimized mEEE against $\mathrm{\Delta T}$ with $W=0.03$ and $m=5$.}\label{OptimzedmEEEVsTheta}
\end{figure}

Fig.~\ref{OptimzedmEEEvsMissionReliability} delineates the optimized mEEE versus target mission reliability. We set the values of $\theta = 10^{-3}$, and $\mathrm{\Delta T} = 5$. The possibility of achieving high mEEE in an ultra-reliable region, i.e., $\sigma > 99.999\%$, is quite low for a single-port system. Generally, systems operating in environments with the highest number of ports tend to have higher mEEE, particularly at high target mission reliability. These insights suggest that system designers should consider that achieving high mission reliability requires energy, which in turn reduces energy efficiency. Hence, industrial applications such as safety-critical systems should be designed with optimized energy to balance the trade-off between mEEE and mission reliability. 

\begin{figure} [!t] 
\centering 
\includegraphics[width=\columnwidth]{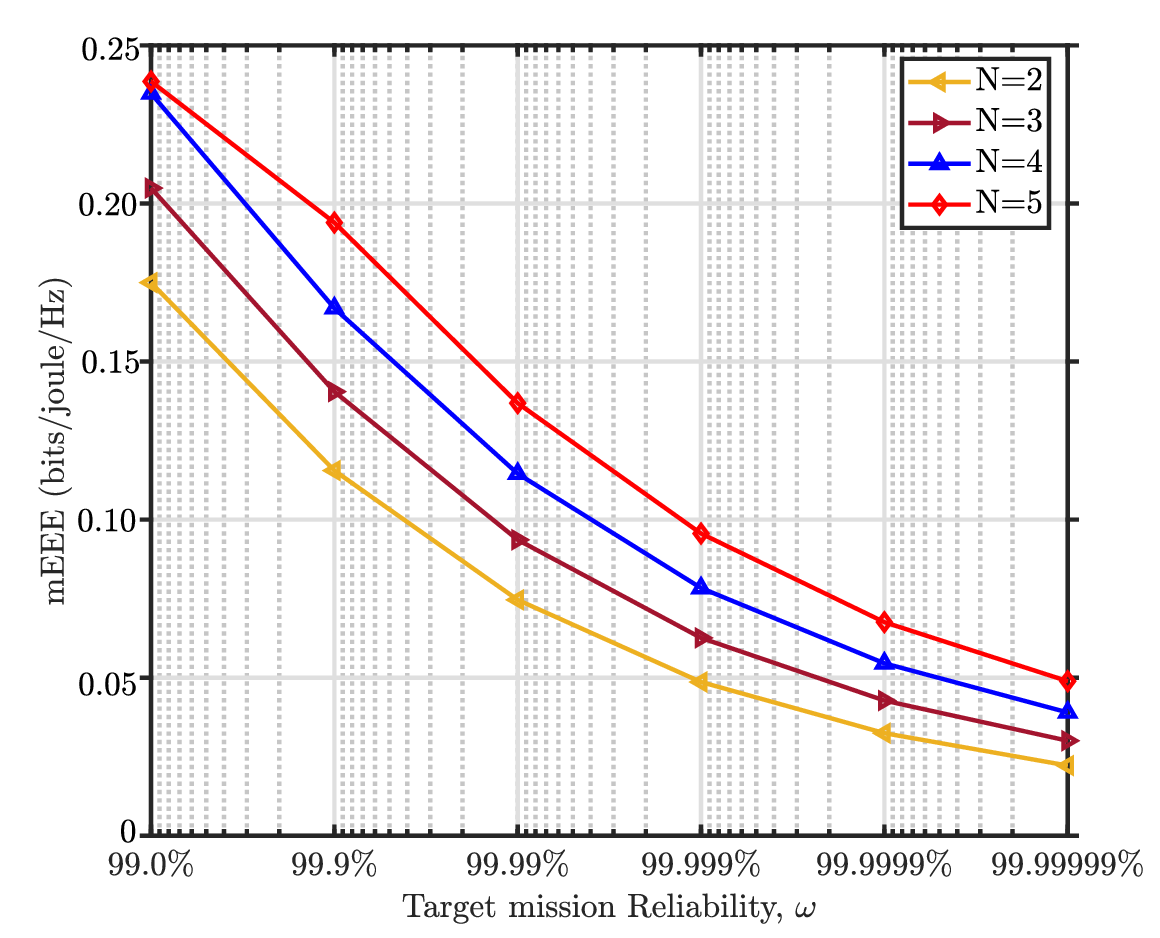}
\caption{The optimized mEEE against $\mathrm{{R}_{M}}$ with $W=0.03$ and $m=4$.}\label{OptimzedmEEEvsMissionReliability}
\end{figure}

\section{Conclusion}\label{sc:conclusion}%
We presented a comprehensive dependability-based analysis of FAS tailored to mission-critical IIoT applications. By deriving closed-form LCR and AFD expressions for an $N$-port FAS under Nakagami-$m$ fading, we revised mission reliability and MTTFF. We extended the EC into the FBL regime, allowing us to define mEC that inherently guarantees failure-free operation over a specified duration. Incorporating realistic traffic models and a refined power-consumption model, we formulated the mEEE metric and developed an efficient modified Dinkelbach algorithm for its maximization. Numerical results illustrated the fundamental trade-offs between reliability, latency, and energy efficiency, demonstrating that increased port diversity and careful SNR optimization are critical to sustaining high mEEE under stringent URLLC requirements. 

\appendices
\section{Proof of Theorem \ref{theo1}}\label{app1}
In this regard, the $N$-variate joint PDF $p_{|\dot{h}|,|h|}(\dot{x},x)$ is given by \cite[(8.42)]{jpdf}
\begin{align}  \label{alou}
p_{|\dot{h}|,|h|}(\dot{x},x)&=\sum_{i=1}^N p_{|\dot{h_i}|}(\dot{x}) \nonumber \\
& \!\! \times \underbrace{\int_0^x\!\!\cdots\!\!\int_0^x}_{(N-1)-{\rm fold}} \!\!\!\! p_{|h_{1}|,\dots,|h_{N}|}(x_1,\dots,x_i=x,\dots,x_N) \nonumber \\
& \times \underbrace{dx_1 \cdots dx_k\cdots dx_N}_{\substack{(N-1)-{\rm fold} \\ k \neq i}},
\end{align}
where $|\dot{h}_i|$ is the time derivative of the signal envelope at the $i$-th port. Hence, from \eqref{lcrdef}, we obtain the LCR as
\begin{align}  \label{step1}
L(x_{\rm th})&=\int_0^{\infty}\dot{x}p_{|\dot{h}|,|h|}(\dot{x},x)d\dot{x} =\int_0^{\infty}\!\!\!\!\dot{x} \sum_{i=1}^N p_{|\dot{h_i}|}(\dot{x})\!\! \nonumber \\
& \times \underbrace{\int_0^{x_{\rm th}}\!\!\!\!\cdots\!\!\int_0^{x_{\rm th}}}_{(N-1)-{\rm fold}}\!\!\!\! \!\!p_{|h_{1}|,\dots,|h_{N}|}(x_1,\dots,x_i=x_{\rm th},\dots,x_N) \nonumber \\
& \times \underbrace{dx_1 \cdots dx_k\cdots dx_N}_{\substack{(N-1)-{\rm fold} \\ k \neq i}}d\dot{x}.
\end{align}
Moreover, in the case of an identically distributed Nakagami-$m$ fading scenario, $p_{|\dot{h_i}|}(\dot{x})$ for $i=1,\dots,N,$ follows a zero mean Gaussian PDF with variance {$\sigma_{\dot{X}}^2=\pi^2 \frac{\sigma^2}{m} f_D^2$, where $f_D$ is the maximum Doppler frequency \cite{stuber}}. Thus, we obtain
\begin{equation}  \label{pdot}
\int_0^{\infty}\!\!\dot{x}p_{|\dot{h_i}|}(\dot{x})d\dot{x}=\frac{\sigma_{\dot{X}}}{\sqrt{2\pi}}=\sqrt{\frac{\pi}{2m}}\sigma  f_D,~i=1,\dots,N.
\end{equation}
By combining \eqref{step1} and \eqref{pdot}, we get
\begin{align}  \label{step2}
&L(x_{\rm th})=\sqrt{\frac{\pi}{2m}}\sigma f_D \nonumber \\
& \times  \sum_{i=1}^N \underbrace{\int_0^{x_{\rm th}}\!\!\cdots\int_0^{x_{\rm th}}}_{(N-1)-{\rm fold}}p_{|h_{1}|,\dots,|h_{N}|}(x_1,\dots,x_i=x_{\rm th},\dots,x_N) \nonumber \\
& \times  \underbrace{dx_1 \cdots dx_k\cdots dx_N}_{\substack{(N-1)-{\rm fold} \\ k \neq i}}.
\end{align}
The summation term in \eqref{step2} is alternatively written as
\begin{align}  \label{altr}
&\sum_{i=1}^N \underbrace{\int_0^{x_{\rm th}}\!\!\cdots\int_0^{x_{\rm th}}}_{(N-1)-{\rm fold}}p_{|h_{1}|,\dots,|h_{N}|}(x_1,\dots,x_i=x_{\rm th},\dots,x_N) \nonumber \\
& \quad \times \underbrace{dx_1 \cdots dx_k\cdots dx_N}_{\substack{(N-1)-{\rm fold} \\ k \neq i}} \nonumber \\
&= \underbrace{\int_0^{x_{\rm th}}\!\!\cdots\int_0^{x_{\rm th}}}_{(N-1)-{\rm fold}}p_{|h_{1}|,\dots,|h_{N}|}(x_1=x_{\rm th},\dots,x_N) \underbrace{dx_2 \cdots dx_N}_{(N-1)-{\rm fold}} \nonumber \\
& \quad + \sum_{i=2}^N \underbrace{\int_0^{x_{\rm th}}\!\!\cdots\!\!\int_0^{x_{\rm th}}}_{(N-1)-{\rm fold}}\!\!p_{|h_{1}|,\dots,|h_{N}|}(x_1,\dots,x_i=x_{\rm th},\dots,x_N) \nonumber \\
& \quad \times \underbrace{dx_1 \cdots dx_k\cdots dx_N}_{\substack{(N-1)-{\rm fold} \\ k \neq i}}.
\end{align}
The first term of \eqref{altr} is evaluated as
\begin{align}  \label{1st}
&\underbrace{\int_0^{x_{\rm th}}\cdots\int_0^{x_{\rm th}}}_{(N-1)-{\rm fold}}p_{|h_{1}|,\dots,|h_{N}|}(x_1=x_{\rm th},\dots,x_N) \underbrace{dx_2 \cdots dx_N}_{(N-1)-{\rm fold}} \nonumber \\
&\overset{(a)}{=}\frac{{2 x_{\rm th}^{2m - 1}{m^m}} {\exp\!\left(\! - \frac{{x_{\rm th}^2m}}{{\sigma^2}}\!\right)  } }{{\Gamma (m)\sigma^{2m}}} \nonumber \\
& \times \int_0^{x_{\rm th}}\!..\int_0^{x_{\rm th}}\!\prod _{k=2}^N
{\frac{{2mx_{\rm th}^{1 - m}x_k^m\sigma^{m - 1}}}{{\sigma^{m + 1}\left( {1 - \mu _k^2} \right)\mu _k^{m - 1}}}} \nonumber \\
& \times \!\exp\!\!\left(\! - \frac{{m\sigma^2 x_k^2 + m\sigma^2 x_{\rm th}^2\mu _k^2}}{{\sigma^4\left( {1 - \mu _k^2} \right)}}\!\right)\nonumber \\
& \times I_{m - 1}\!\!\left[ {\frac{{2m{\mu _k}{x_{\rm th}}{x_k}}}{{{\sigma^2}\left( {1 - \mu _k^2} \right)}}} \!\right]\!\! dx_2 \cdots dx_N \nonumber \\
& =\frac{{2x_{\rm th}^{2m-1}{m^m}}{\exp\!\left( - \frac{{x_{\rm th}^2m}}{{\sigma^2}}  \right)   }}{{\Gamma (m)\sigma^{2m}}}\! \prod _{k=2}^N \int_0^{x_{\rm th}}\!\!{\frac{{2mx_{\rm th}^{1 - m}x_k^m\sigma^{m - 1}}}{{\sigma^{m + 1}\left( {1 - \mu _k^2} \right)\mu _k^{m - 1}}}} \nonumber \\
& \quad \times{\exp\!\!\left(\! - \frac{{m\sigma^2 x_k^2 + m\sigma _k^2x_{\rm th}^2\mu _k^2}}{{\sigma^4\left( {1 - \mu _k^2} \right)}}  \right)    }  {I_{m - 1}}\left[ {\frac{{2m{\mu _k}{x_{\rm th}}{x_k}}}{{{\sigma^2}\left( {1 - \mu _k^2} \right)}}} \right] dx_k \nonumber \\
&\overset{(b)}{=}\!\!\frac{{2x_{\rm th}^{2m-\! 1}} {\exp\!\left(\! - \frac{{x_{\rm th}^2m}}{{\sigma^2}}\! \right)  }  }{{{m^{-m}}\Gamma (m)\sigma^{2m}}}\nonumber \\
& \quad \times\prod _{k=2}^N \!\!{\left[\! {1\!\!-\!{Q_{\text{m}}}\!\!\left(\!\! {\sqrt{\!\frac{{2m\mu _k^2x_{\rm th}^2}}{{\sigma^2\left( {1 - \mu _k^2} \right)}}} \!,\!
\sqrt{\!\frac{{2m} x_{\rm th}^2}{{\sigma^2\left( {1 - \mu _k^2}\right)}}}  } \right) }\!\!\right]},
\end{align}
where $(a)$ follows from \eqref{pdf_FAS_Nakagami} and $(b)$ follows from \cite[(10)]{marcum}. Hereafter, the second term of \eqref{altr} is expanded as \eqref{2nd} (see top of next page), where $(c)$ is based on \cite[(10)]{marcum} and the fact that the joint PDF in \eqref{pdf_FAS_Nakagami} is not a regular multivariate Nakagami-$m$ PDF. This distribution is a product of $N$ pairs of bivariate Nakagami-$m$ PDFs, with the first port of the FAS serving as the reference point for all the remaining  ports. By combining \eqref{step2}--\eqref{2nd}, we obtain \eqref{lcrp}, which completes the proof.
\begin{figure*}[h]
\begin{align}  \label{2nd}
& \sum_{i=2}^N \underbrace{\int_0^{x_{\rm th}}\!\!\cdots\int_0^{x_{\rm th}}}_{(N-1)-{\rm fold}}p_{|h_{1}|,\dots,|h_{N}|}(x_1,\dots,x_i=x_{\rm th},\dots,x_N) \underbrace{dx_1 \cdots dx_k\cdots dx_N}_{\substack{(N-1)-{\rm fold} \\ k \neq i}} \nonumber \\
&=\!\! \sum_{i=2}^N\!\int_0^{x_{\rm th}}\!\!\cdots\!\!\int_0^{x_{\rm th}} \frac{2mx_1^{1 - m}x_{\rm th}^m\sigma^{m - 1} \exp\!\!
\left(\!\! - \frac{m\sigma^2 x_{\rm th}^2 + m\sigma^2 x_1^2\mu _k^2}{\sigma^4\left( {1 - \mu _i^2} \right)}\!\!\right)\!\!  }{\sigma^{m+1}\left( 1 - \mu _i^2 \right)\mu _i^{m - 1}}  I_{m-1}\!\!\left[\! \frac{2m \mu_k x_1 x_{\rm th}}{\sigma^2\left( 1-\mu_i^2 \right)}\! \right]\!\!\prod _{\substack{k=1 \\ k \neq i}}^N \frac{2mx_1^{1 - m}x_k^m\sigma^{m - 1}}{\sigma^{m+1}\left( 1 - \mu _k^2 \right)\mu _k^{m - 1}}   \nonumber \\
& \times \exp\!\left( - \frac{m\sigma^2 x_k^2 + m\sigma^2 x_1^2\mu _k^2}{\sigma^4\left( {1 - \mu _k^2} \right)}\right) I_{m-1} \left[ \frac{2m \mu_k x_1 x_k }{\sigma^2\left( 1-\mu_k^2 \right)} \right]  dx_1.. dx_k.. dx_N  \nonumber \\
&\overset{(c)}{=} \! \sum_{i=2}^N \frac{2m x_{\rm th}^m\sigma^{m - 1}}{\sigma^{m+1}\left( 1 - \mu _i^2 \right)\mu _i^{m - 1}} \exp\!\!\left(\!\!- \frac{mx_{\rm th}^2 }{\sigma^2\left( {1 - \mu _i^2} \right)}\right) \!\!\! \int_0^{x_{\rm th}} \! x_1^{1 - m} \frac{2x_1^{2m - 1} m^m}{\Gamma (m)\sigma^{2m}} \exp\!\!\left(\!\!- \frac{mx_1^2}{\sigma^2}\right) \exp\!\!\left(\! - \frac{m x_1^2\mu _i^2}{\sigma^2\left( {1 - \mu _i^2} \right)}\right)I_{m-1}\left[ \frac{2m \mu_i x_1 x_{\rm th}}{\sigma^2\left( 1-\mu_i^2 \right)} \right] \nonumber \\
& \times \prod _{\substack{k=2 \\ k \neq i}}^N \int_0^{x_{\rm th}} \frac{2mx_1^{1 - m}x_k^m\sigma^{m - 1}}{\sigma^{m+1}\left( 1 - \mu _k^2 \right)\mu _k^{m - 1}}   \exp\!\!\left(\!\! - \frac{m\sigma^2 x_k^2 + m\sigma^2 x_1^2\mu _k^2}{\sigma^4\left( {1 - \mu _k^2} \right)}\right) I_{m-1} \left[ \frac{2m \mu_k x_1 x_k}{\sigma^2\left( 1-\mu_k^2 \right)} \right]dx_1 \nonumber \\
&=\sum_{i=2}^N \frac{2m x_{\rm th}^m\sigma^{m - 1}}{\sigma^{m+1}\left( 1 - \mu _i^2 \right)\mu _i^{m - 1}} \exp\!\!\left(\!\! - \frac{mx_{\rm th}^2 }{\sigma^2\left( {1 - \mu _i^2} \right)}\right) \int_0^{x_{\rm th}}\frac{2x_1^{3m - 2} m^m}{\Gamma (m)\sigma^{2m}}  \exp\!\!\left(\!\! - \frac{m x_1^2}{\sigma^2\left( {1 - \mu _i^2} \right)}\right) I_{m-1}\left[ \frac{2m \mu_i x_1 x_{\rm th}}{\sigma^2\left( 1-\mu_i^2 \right)} \right] \nonumber \\
& \times  \prod _{\substack{k=2 \\ k \neq i}}^N {\left[ {1 - {Q_{\text{m}}}\left( {\sqrt {\frac{{2m\mu _k^2 }}{{\sigma^2\left( {1 - \mu _k^2} \right)}}}x_1 ,\sqrt {\frac{{2m}}{{\sigma^2\left( {1 - \mu _k^2} \right)}}} x_{\rm th}} \right)} \right]} dx_1
\end{align}
\hrule
\end{figure*}

\section{Proof of Corollary \ref{mu1}}\label{appmu1}
The case of $\mu_k=1~\forall k$ implies identical channels at all the ports. As a result, we obtain
\begin{align}
L(x_{\rm th})&=\int_0^{\infty} \!\! \dot{x}p_{|\dot{h}|,|h|}(\dot{x},x_{\rm th})d\dot{x} \nonumber \\
&\overset{(a)}{=}p_{|h|}(x_{\rm th})\!\! \int_0^{\infty} \dot{x}p_{|\dot{h}|}(\dot{x})d\dot{x} \nonumber \\
&=\frac{2m^{m} {x_{\rm th}}^{2m-1}}{\Gamma(m)\sigma^2}   \exp\left(-\frac{m x_{\rm th}^2}{\sigma^2}\right)   \!\!\int_0^{\infty}\!\!\dot{x}p_{|\dot{h_i}|}(\dot{x})d\dot{x} \nonumber \\
&\overset{(b)}{=}\frac{\sqrt{2\pi} f_D m^{m-\frac{1}{2}}{x_{\rm th}}^{2m-1}  \exp\left(-\frac{m x_{\rm th}^2}{\sigma^2}\right)  }{\Gamma(m)\sigma^{2m-1}},
\end{align}
\noindent where $(a)$ follows from $p_{|\dot{h}|,|h|}(\dot{x},x)=p_{|\dot{h}|}(\dot{x})p_{|h|}(x)$ \cite[(2.97)]{stuber} and $(b)$ follows from \eqref{pdot}.

\section{Proof of Corollary \ref{corr2}}  \label{app2}
Since we are considering a two-port scenario, we take into account the joint PDF of the channel at these two ports. Hence, by replacing $N=2$ in \eqref{pdf_FAS_Nakagami}, the joint PDF becomes
\begin{align}  \label{pdf2var_nakagami} 
{p_{\left| {{h_1}} \right|,\left| {{h_2}} \right|}}&\left( {{x_1},{x_2}} \right) = \frac{{4{m^{m + 1}}{{\left( {{x_1}{x_2}} \right)}^m}{\exp\left( - \frac{m}{{1 - \mu _2^2}} \left( {\frac{{x_1^2}}{{\sigma^2}} + \frac{{x_2^2}}{{\sigma^2}}} \right) \right)   }}}{{\Gamma (m)\sigma^4\left( {1 - \mu _2^2} \right){{\left( {  {\sigma^2}   {\mu _2}} \right)}^{m - 1}}}} \nonumber \\
& \!\!\!\!\!\!\!\!\!\!\! \times {I_{m - 1}}\left( {\frac{{2m{\mu_2}{x_1}{x_2}}}{{\left( {1 - {u_2^2}} \right){\sigma^2}   }}} \right),~\mbox{for }x_1,x_2 \geq 0.
\end{align}
By replacing $N=2$ in~\eqref{lcrp} and after some trivial algebraic manipulations, we obtain
\begin{align}  \label{n21}
L(x_{\rm th})&=\sqrt{\frac{\pi}{2m}}\sigma f_D \left( \int_0^{x_{\rm th}} p_{|h_1|,|h_2|}(x_{\rm th},x_2)dx_2 \right. \nonumber \\
&\left. +\int_0^{x_{\rm th}} p_{|h_1|,|h_2|}(x_1,x_{\rm th})dx_1 \right) \nonumber \\
&\overset{(a)}{=}\sqrt{\frac{2\pi}{m}}\sigma f_D \!\! \int_0^{x_{\rm th}} p_{|h_1|,|h_2|}(x_{\rm th},x_2)dx_2 \nonumber \\
& = \frac{\sqrt{2\pi } 4 m^{m + 1/2} f_D x_{\rm th}^m \exp \left(\!\!-\frac {m x_{\rm th}^2}{\sigma^{2}(1-\mu_2^2)} \right)}{ \Gamma (m)\sigma^{2m+1}  (1-\mu_2^{2}) \mu_2^{m-1}}  \nonumber \\
& \times \!\int_0^{x_{\rm th}}\!\! x_2^m \exp \left(\!\!-\frac {m x_2^2}{\sigma^{2}(1-\mu_2^2)}\!\! \right) \!\!{I_{m - 1}}\left ( \!{\frac {2m\mu x_{\rm th}x_2}{\sigma^2(1-\mu_2^2)}}\!\right)\! dx_2, 
\end{align}
where $(a)$ follows from~\eqref{pdf2var_nakagami} and
\begin{align}  \label{n22}
&\int_0^{x_{\rm th}}\!\! x_2^m \! \exp \left(\!\!-\frac {m x_2^2}{\sigma^{2}(1-\mu_2^2)}\!\! \right) \!\!{I_{m - 1}}\left ( \!{\frac {2m\mu x_{\rm th}x_2}{\sigma^2(1-\mu_2^2)}}\!\right)\! dx_2 \nonumber \\
&\overset{(b)}{=}\!\! \!\! \int_0^{x_{\rm th}}\! \!\! x_2^m \! \exp\!\! \left(\!\!-\frac {m x_2^2}{\sigma^{2}(\!1-\!\mu_2^2)}\!\! \right) \nonumber \\
& \times \sum_{k=0}^{\infty} \frac{1}{k! \Gamma(m\!+\!k)}\!\!\left(\!\!\frac {m\mu x_{\rm th}x_2}{\sigma^2(1\!-\!\mu_2^2)} \!\!\right)^{\!\! 2k+m-1} \!\! dx_2 \nonumber \\
&\overset{(c)}{=} \sum_{k=0}^{\infty} \frac{1}{(k!)\Gamma(m+k)} \left(\frac {m\mu_2 x_{\rm th}}{\sigma^2(1-\mu_2^2)} \right)^{2k+m-1} \nonumber \\
& \times \int_0^{x_{\rm th}} \exp \left(-\frac {m x_2^2}{\sigma^{2}(1-\mu_2^2)} \right) x_2^{2k+2m-1} dx_2 \nonumber \\
&=\frac{1}{2}\sum_{k=0}^{\infty}\frac{\left(\mu_2 x_{\rm th}\right)^{2k+m-\!1}}{(k!)\Gamma(m+k) } \left(\frac{ m }{\sigma^2(1-\mu_2^2)}\right)^{k-1}\nonumber \\
& \times \gamma \left(k+m, \frac{m x_{\rm th}^2}{\sigma^2 (1-mu_2^2)}\right),
\end{align}
where $(b)$ follows from \cite[(8.445)]{grad} and by assuming $0<\mu<1$, $(c)$ follows from changing the order of summation and integration. As a result, by substituting \eqref{n22} into \eqref{n21}, we finally obtain \eqref{n2f}, which completes the proof.

\bibliographystyle{IEEEtran}


\end{document}